\documentclass{article}
\usepackage{amsmath, amsthm, amsfonts, amssymb, amscd, bm, enumerate}
\AtBeginDocument{%
  \paperwidth=\dimexpr
    1in + \oddsidemargin
    + \textwidth
    + 1in + \oddsidemargin
  \relax
  \paperheight=\dimexpr
    1in + \topmargin
    + \headheight + \headsep
    + \textheight
    + 1in + \topmargin
  \relax
  \usepackage[pass]{geometry}\relax
}

\usepackage{amsmath, amsfonts, amssymb, mathrsfs}
\DeclareMathOperator*{\argmin}{arg\,min}
 
\usepackage{bm} 
\usepackage{cool} 
\usepackage{mathtools}
\usepackage{bbm}
\makeatletter 
\providecommand*{\diff}%
	{\@ifnextchar^{\DIfF}{\DIfF^{}}}
\def\DIfF^#1{%
	\mathop{\mathrm{\mathstrut d}}%
	\nolimits^{#1}\gobblespace}
\def\gobblespace{%
	\futurelet\diffarg\opspace}
\def\opspace{%
	\let\DiffSpace\!%
	\ifx\diffarg(%
		\let\DiffSpace\relax
	\else
		\ifx\diffarg[%
			\let\DiffSpace\relax
		\else
			\ifx\diffarg\{%
				\let\DiffSpace\relax
			\fi\fi\fi\DiffSpace}

\usepackage{xcolor}
\definecolor{oxford_blue}{RGB}{14,31,71}
\usepackage{graphicx}
\usepackage{wrapfig}
\usepackage{booktabs}
\usepackage{multirow}
\usepackage[labelfont=bf]{caption} 
\usepackage{subcaption} 
\usepackage[ruled,vlined,linesnumbered,nofillcomment]{algorithm2e}

\SetCommentSty{mycommfont}
\SetKw{KwBy}{by}

\usepackage{enumerate, enumitem}
\usepackage[colorlinks]{hyperref}
\usepackage{etoolbox} 
\usepackage{url}
\usepackage{tikz}
\usepackage{tkz-graph}
\usetikzlibrary{shapes,intersections,calc,arrows,arrows.meta,positioning,fit,chains,decorations.markings}
\tikzset{%
  >={stealth[width=1mm,length=1mm]},
            base/.style = {rectangle, rounded corners, draw=black,
                           minimum width=1.5cm, minimum height=1cm,
                           text centered, font=\sffamily},
  			mass/.style = {base, text width=3cm, thick},
         process/.style = {base, on chain,
         				   minimum width=1.2cm, fill=white,
         				   text width=1.2cm},
            data/.style = {trapezium, draw, on chain,
            			   anchor=center,
            			   minimum width=1.5cm, minimum height=.75cm,
            			   trapezium left angle=75, trapezium right angle=105,
            			   text centered, font=\ttfamily,
            			   text width=1.6cm,
            			   fill=red!30},
           	  or/.style = {draw, circle, minimum size=0.05cm, fill=blue!30},
}

\newtheorem{definition}{Definition}
\newtheorem{proposition}{Proposition}
\newtheorem{lemma}{Lemma}

\begin{document}

\title{Detecting and repairing arbitrage in traded option prices
}



\author{Samuel N. Cohen \and Christoph Reisinger \and Sheng Wang \\
Mathematical Institute, University of Oxford \\
\texttt{ \{samuel.cohen, christoph.reisinger, sheng.wang\} }\\ \texttt{@maths.ox.ac.uk}
}


%

\maketitle

\begin{abstract}

Option price data are used as inputs for model calibration, risk-neutral density estimation and many other financial applications. The presence of arbitrage in option price data can lead to poor performance or even failure of these tasks, making pre-processing of the data to eliminate arbitrage necessary. Most attention in the relevant literature has been devoted to arbitrage-free smoothing and filtering (i.e. removing) of data. In contrast to smoothing, which typically changes nearly all data, or filtering, which truncates data, we propose to repair data by only necessary and minimal changes. We formulate the data repair as a linear programming (LP) problem, where the no-arbitrage relations are constraints, and the objective is to minimise prices' changes within their bid and ask price bounds. Through empirical studies, we show that the proposed arbitrage repair method gives sparse perturbations on data, and is fast when applied to real world large-scale problems due to the LP formulation. In addition, we show that removing arbitrage from prices data by our repair method can improve model calibration with enhanced robustness and reduced calibration error.

\end{abstract}

 
\newpage

\section{Introduction}

Price data of vanilla options are widely used in various financial applications such as calibrating models for pricing and hedging, and computing risk-neutral densities (RND) of the underlying. The presence of arbitrage in option price data can lead to poor or even failed model calibration, as well as erroneous RND estimation. Derivative pricing models in nearly all applications are constructed to be arbitrage-free, as it is economically meaningless to have a model that has the potential to make risk-free profits. Exact calibration is impossible for any arbitrage-free model or RND function with any size of input data if the price data contain arbitrage. For example, calibration of the local volatility model of Dupire \cite{Dupire1994} and Derman and Kani \cite{Derman1994} would fail given arbitrageable data. It is also not possible to have any arbitrage-free interpolation such as Kahale \cite{Kahale2003} and Wang et al. \cite{Wang2004}, because the data to be interpolated are not arbitrage-free from the beginning. Though inexact calibration methods are available for many models, it seems natural to expect that removing arbitrage from input data can improve calibration of arbitrage-free models, such as enhancing robustness or reducing calibration error.

Therefore, it is important to remove arbitrage (if present) from option price data. Most attention in the relevant literature has been devoted to the \textit{smoothing} and \textit{filtering} of data. Notable works on smoothing include A\"{i}t-Sahalia and Duarte \cite{AitSahalia2003}, Fengler et al. \cite{Fengler2009} \cite{Fengler2012} \cite{Fengler2015}, Gatheral and Jacquier \cite{Gatheral2014}, and Lim \cite{lim2020improved}. In fact, the calibration of many pricing models, such as stochastic volatility models, is essentially arbitrage-free smoothing. Arbitrage-free data is only a byproduct of smoothing, since the main goal of smoothing is to produce a $C^{1,2}$ call price function $(T,K) \mapsto C(T,K)$ (or equivalently implied volatility function $(T,K) \mapsto \sigma_\text{imp}(T,K)$). For smoothing, usually an $\ell^2$-norm optimisation is used when searching over polynomial, spline or kernel parameters that produce values as close to the given price data as possible. This method leads to changes for nearly \textit{all} data. Though liquidity considerations can be included by adding weights in the optimisation, it remains unclear what should be an effective way to set weight values for different options. The filtering of data refers to simply removing suspiciously low-quality price data according to criteria in terms of moneyness, expiry, trading volume, intra-day activity, etc. A good survey of popular empirical filtering criteria can be found in Ivanovas \cite{Ivanovas2015} and Meier \cite{Meier2015}. Filtering can be quite subjective, can cause information loss, and might not even be feasible as many criteria are based on order-book level data, which are not always available (for example, in OTC markets).

In contrast to smoothing, which typically changes nearly all data, or filtering, which truncates data, we propose to \textit{repair} data in the sense that only necessary and minimal changes are made to the given data in the presence of arbitrage. If arbitrages in data are mainly consequences of infrequent price updates of illiquid options rather than noncompetitive market, it is better to only perturb as few data points as possible. In addition, when making changes, we use bid and ask prices as soft bounds such that liquidity profiles of different options are considered in an objective way. Bid-ask spread is a measure of liquidity, i.e. the narrower the spread, the easier a market order can be matched and executed. Since the ``fair'' price could lie anywhere within the bid-ask price bounds, the width of the bid-ask spread represents the degree of certainty in the market prices. Empirically, deep out-of-the-money (OTM) and in-the-money (ITM) options are thinly traded with wide bid-ask spreads, leading to less trustworthy price data compared with more liquid options. We therefore formulate the data repair as a constrained optimisation problem, where the no-arbitrage relations are written as constraints, and the objective is to minimise price changes within soft bounds. By carefully choosing the objective function, we can rewrite the formulation as a linear programming (LP) problem, so that we can take advantage of efficient solution techniques and software for large-scale LPs.

Our method is to repair single-price data. At any moment during the trading day, each tradable asset has multiple prices, i.e. bid price and ask price. However, most applications require single-price inputs. There is a need to construct some ``fair'' reference price from the market-quoted multiple prices. Examples of a reference price are the mid-price, the quantity-weighted price, the last trade price or the micro-price by Stoikov \cite{Stoikov2018}. In this article, we do not discuss the construction of reference price, and use the mid-price by default, however other reference prices could easily be considered.

We envisage further applications of our methodology in repairing data generated by models which do not themselves rule out arbitrage. Included in this class are prices predicted by deep learning methods, which have gained substantial popularity recently, as documented by the survey paper by Ruf and Wang \cite{ruf2019neural}. Typically, there is no guarantee for arbitrage-free predicted option prices even if the training set is arbitrage free; see also a more detailed discussion of this point in the introduction of Dixon, Cr\'{e}pey and Chataigner \cite{dixon2020deep}, which goes on to use the local volatility code book for arbitrage free vanilla prices as a means of guaranteeing arbitrage-free interpolation of prices. The arbitrage repair method from our paper can provide a simple post-processing step of potentially arbitrageable learned prices. By repairing a discrete set of input prices directly without extra assumptions, using linear constraints only, the method distinguishes itself by versatility, transparency, and speed, making it particularly well-suited to online computations.

The rest of the paper is structured as follows.  We derive a set of empirically verifiable model-independent, static arbitrage constraints in Section \ref{sec:arbitrage}. Our derivation is mainly based on Carr, G\'{e}man, Madan and Yor \cite{CGMY2003} \cite{Carr2005}, Davis and Hobson \cite{davis2007},  and Cousot \cite{cousot2007} \cite{cousot2004}. In Section \ref{sec:arbitrage_repair}, we formulate data repair as a constrained LP problem, and the design of the objective function is carefully discussed\footnote{Our implementation of this algorithm in Python is available in the repository \url{https://github.com/vicaws/arbitragerepair}.}. Finally in Section \ref{sec:numerics}, we apply our arbitrage repair method to FX option data to justify why arbitrage repair is needed for real data, and demonstrate how our method performs empirically on various metrics, especially on the improvement of model calibration. We also show an example of how we can use our repair method for identifying the formation and disappearance of executable arbitrage in intra-day price data.

\section{Arbitrage constraints}
\label{sec:arbitrage}

We consider a finite collection of traded European call options\footnote{We focus on European style vanilla options in this study. Specifically, we only consider call options, since the static arbitrage constraint between call and put options is the put-call parity, which can be easily incorporated in our approach. The framework of our arbitrage repair method is applicable to a mixture of a wider range of options, as long as their arbitrage constraints can be defined by feasible linear inequalities of prices.} written on the same asset. These options can have arbitrary expiry and strike parameters rather than a rectangular grid of parameters, a restrictive prerequisite for many arbitrage detection \cite{Carr2005} and spline-type smoothing methods \cite{Fengler2009} \cite{Kahale2003} to work. In practice, it is uncommon to have price data on a rectangular grid, see, e.g. Figure \ref{fig:optionLattice}.

Consider $N$ European call options that have expiries $0 < T_1 < T_2 < \cdots < T_m$. For a given expiry $T_i$, available strikes are $0 < K_1^i < K_2^i < \cdots < K_{n_i}^i$. The $(i,j)$-th option has reference price $C^i_{j}$ at present time $0$, and terminal payoff $(S_{T_i} - K^i_j)^+$, where $S_t$ denotes the price of the underlying asset at time $t$. Hence $N = \sum_{i=1}^m n_i$. Denote $\mathcal{T}^e = \{T_i\}_{1 \leq i \leq m} $ and $\mathcal{P}^{T,K} = \{(T_i, K^i_j)\}_{1 \leq j \leq n_i, 1 \leq i \leq m}$. In Figure \ref{fig:optionLattice}, we show how $(T,K)$ are distributing for traded call options on a typical trading day. A detailed description of the data used can be seen in Section \ref{sec:numerics}.
\begin{figure}[!h]
\centering
\includegraphics[scale=.59]{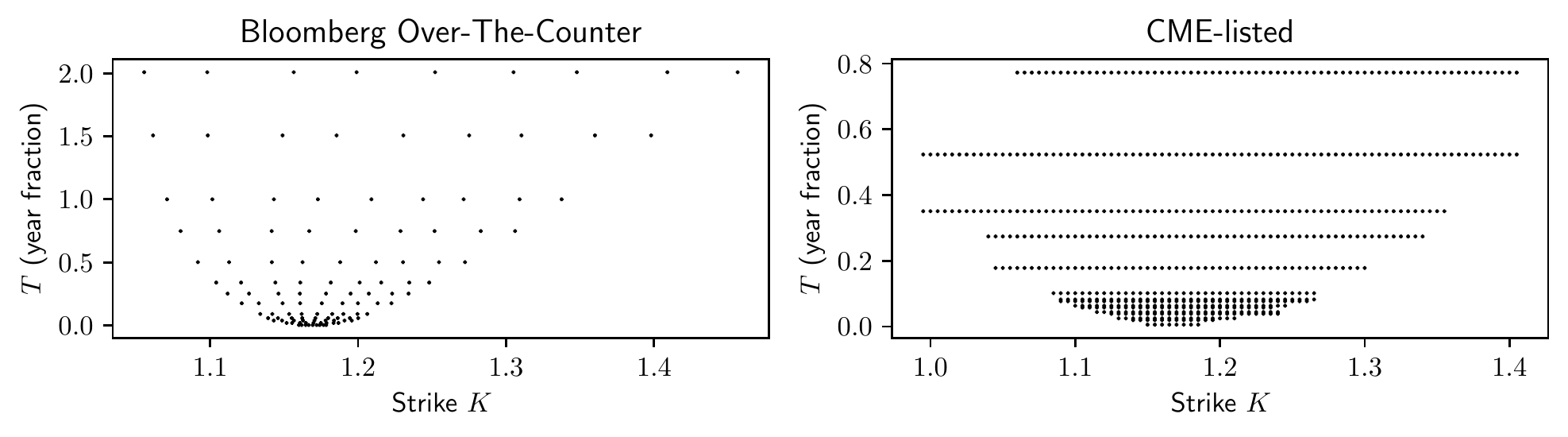}
\caption{Distributions of $(K,T)$ for traded EURUSD call options in the OTC market (Bloomberg data) and at the CME market, observed as of 31st May, 2018.}
\label{fig:optionLattice}
\end{figure}

\subsection{Assumptions}

The arbitrage constraints that price data should satisfy are derived under a frictionless market assumption. As a consequence, when the price data break these constraints, it may not be possible in practice to exploit the apparent arbitrage, given practical market barriers and transaction costs. However, the assumption that prices should be arbitrage free is justified by the fact that the single-price data are not executable prices in the market, but are designed to be reference or benchmark prices for tradable assets, which are useful inputs to a variety of models. Nevertheless, we will use bid and ask prices as soft bounds for guiding our arbitrage repair mechanism.

We allow for non-zero but deterministic interest and dividends\footnote{When applying our method to other asset classes, dividends of stock shares are comparable to foreign currency interest rates for FX rates, or convenience yields for commodities.}. At present time 0, we use $D(T)$ to denote the market discount factor for time $T$, and $\Gamma(T)$ to denote the number of shares which will be owned by time $T$ if dividend income is invested in shares. Then there is a model-independent, arbitrage-free forward price, $F(T) = S_0 / (\Gamma(T) D(T))$, for delivery of the asset at $T$.

We assume that zero-coupon bonds and forward contracts on the risky asset, with the same expiries as the options, are traded in the market. In addition, they are sufficiently liquid that we can neglect their bid-ask spreads (e.g. usually one or two ticks). Therefore, we observe market discount factors $D_i:=D(T_i)$ and forward prices $F_i:=F(T_i)$ for $1 \leq i \leq m$. However, when the underlying (spot or forward) trades at a sufficiently large bid-ask spread, then any arbitrage strategy can become impossible (see the discussion by Gerhold and G\"{u}l\"{u}m \cite{Gerhold2020}).

\subsection{Static arbitrage}

Arbitrage refers to a costless trading strategy that has a positive probability of earning risk-free profit. A \textit{static arbitrage} is an arbitrage exploitable by fixed positions in options and the underlying stock at initial time, while the position of the underlying stock can be modified at a finite number of trading times in the future. Any other arbitrage is called \textit{dynamic arbitrage}. As an example of a static arbitrage, it must hold the condition that $C^1_1 \geq C^1_2$ for $K^1_1 < K^1_2$, otherwise by going long one $(T_1,K^1_1)$ option and short one $(T_1,K^1_2)$ option, we make immediate profit of $C^1_2 - C^1_1$ with non-negative terminal payoff. An example of dynamic arbitrage is a continuously delta-hedged short position on an over-priced option in the perfect Black--Scholes world.

Dynamic arbitrage relies on dynamics and path properties of the tradable assets. From the data repair perspective, we should minimise model dependence, because the repaired data are to be used in more generic applications. Hence, data should only be adjusted by model-independent constraints, so we restrict ourselves to static arbitrage in which no dynamics need to be modelled. Static arbitrage constraints establish the prerequisites that the price data have to satisfy at time zero for admitting a dynamically arbitrage-free model.

A \textit{model} $\mathbb{M}$ is a filtered probability space $(\Omega$, $\mathscr{F}$, $\{\mathscr{F}_t\}_{t \in \mathcal{T}}$, $\mathbb{P})$, that carries an adapted price process $\{(S_t, \mathbf{C}_t)\}_{t \in \mathcal{T}}$, where $\mathbf{C}_t$ gives the prices of the $N$ options at time $t$, and we observe $\mathbf{C}_0$. Here $\mathcal{T}$ denotes the set of times at which the asset can be traded so that $0 \in \mathcal{T}$ and $\mathcal{T}^e \subset \mathcal{T}$, and $\mathscr{F}_0 = \{\Omega, \emptyset\}$ augmented with all null sets of $\mathscr{F}_{T_m}$. 

The First Fundamental Theorem of Asset Pricing (FFTAP) establishes an equivalence relation between no-arbitrage (static and dynamic) and the existence of an equivalent martingale measure (EMM). After the landmark work of Harrison and Kreps \cite{harrison1979}, there are various versions of the FFTAP and extensions of the no-arbitrage concept (e.g. no free lunch by Kreps \cite{Kreps1981}, no free lunch with vanishing risk by Delbaen and Schachermayer \cite{Delbaen1994}). In this article, we work with a simplified version of FFTAP as follows. Given a model $\mathbb{M}$, there is no arbitrage if and only if $\exists \mathbb{Q} \sim \mathbb{P}$, such that
\begin{equation}
    \forall (T,K) \in \mathcal{P}^{T,K} \cup ( \mathcal{T}^e \times \{0\} ), ~D(t) C_t(T,K) =  D(s) \mathbb{E}^\mathbb{Q} \left[ C_s(T,K) | \mathscr{F}_t \right]
\label{eq:martingale_asset}
\end{equation}
for all $t < s \leq T$ where $t,s \in \mathcal{T}$. No static arbitrage corresponds to a much smaller set of conditions, since the path dynamics governed by $\mathbb{Q}$ no longer matter. As discussed by Carr, G\'{e}man, Madan and Yor \cite{CGMY2003}, \cite{Carr2005} and Davis \cite{davis2007}, static arbitrage is present if no $\mathbb{Q}$ exists such that $C_0(T,K) =  D(T) \mathbb{E}^\mathbb{Q} [C_T(T,K) | \mathscr{F}_0]$. Therefore, static arbitrage constraints are consequences of relations between terminal payoffs, projected to the present time.

\subsection{Shape constraints of the call price surface}

Let us define $M_{T_i} = S_{T_i} / F_i$, $k^i_j = K^i_j/F_i$, $c^i_j = C^i_j / (D_i F_i)$, for all $i,j$. To have no static arbitrage, there must exist $\mathbb{Q}$ such that
\begin{equation*}
c^i_j = \mathbb{E}^\mathbb{Q} \left[ \left. \left( \frac{S_{T_i}}{F_i} - \frac{K^i_j}{F_i} \right)^+ \right| \mathscr{F}_0 \right] = \mathbb{E}^\mathbb{Q} \left[ \left. \left( M_{T_i} - k^i_j \right)^+ \right| \mathscr{F}_0 \right],  \quad \forall i,j.
\label{eq:no_arbitrage}
\end{equation*}
We will work on ``normalised'' quantities $M$, $k$, $c$ in the rest of this section. We define the \textit{normalised call function} $c(T,k)$ as
\begin{equation}
c(T, k) := \mathbb{E}^\mathbb{Q} \left[ \left. \left( M_{T} - k \right)^+ \right| \mathscr{F}_0 \right], \text{ where } T \in \mathbb{R}_{>0}, ~k \in \mathbb{R}_{\geq 0}. 
\label{eq:call_function}
\end{equation}

Given the specific structure in (\ref{eq:call_function}), a probability measure $\mathbb{Q}$ exists only when the call function satisfies some shape constraints. For arbitrary but fixed $T$, using Breeden and Litzenberger's analysis \cite{breeden1978}, the marginal measure $\mathbb{Q}_T := \mathbb{Q}(\cdot|\mathscr{F}_T)$ exists if
\begin{equation*}
    \forall~ k_3 > k_2 > k_1 \geq 0, ~ -1 \leq \frac{c(T,k_2)-c(T,k_1)}{k_2-k_1} \leq \frac{c(T,k_3)-c(T,k_2)}{k_3-k_2} \leq 0.
\end{equation*}
If a family of marginal measures $\{\mathbb{Q}_T\}_{T \in \mathcal{T}^e}$ on $(\mathbb{R}, \mathcal{B}(\mathbb{R}))$ exists with time-independent mean, and $\mathbb{Q}_{T_1} \geq_\text{cx} \mathbb{Q}_{T_2}$ whenever $T_1 > T_2$, then, by Kellerer's theorem \cite{kellerer1972}, there exists a Markov martingale measure with these marginals. Here we write $\mathbb{Q}_{T_1} \geq_\text{cx} \mathbb{Q}_{T_2}$ if $\int_\mathbb{R} \phi \diff \mathbb{Q}_{T_1} \geq \int_\mathbb{R} \phi \diff \mathbb{Q}_{T_2}$ for each convex function $\phi:\mathbb{R} \rightarrow \mathbb{R}$, and we say $\{\mathbb{Q}_T\}_{T \in \mathcal{T}^e}$ is Non-Decreasing in Convex Order (NDCO). The convex order can be equivalently characterised in terms of the call function \cite{shaked2007}:
\begin{equation*}
\mathbb{Q}_{T_1} \geq_\text{cx} \mathbb{Q}_{T_2} \Longleftrightarrow
\begin{dcases}
& \mathbb{Q}_{T_i} \text{ and } \mathbb{Q}_{T_j} \text{ have equal means}; \\
& \int_\mathbb{R} (x-k)^+ ~\text{d}\mathbb{Q}_{T_1} \geq \int_\mathbb{R} (x-k)^+ ~\text{d}\mathbb{Q}_{T_2} \quad \forall x\in \mathbb{R}.
\end{dcases}
\end{equation*}
Given that $\mathbb{E}^{\mathbb{Q}_T} [M_U] = \mathbb{E}^{\mathbb{Q}_T} [S_U / F_T(U)] = 1$ is time-independent for any $T < U $ where $T,U \in \mathcal{T}^e$, it is then sufficient to conclude that $\{\mathbb{Q}_T\}_{T \in \mathcal{T}^e}$ is NDCO if $c(T,\cdot) \leq c(U,\cdot)$. Also note that $\lim_{k \downarrow 0} c(T,k) = \mathbb{E}^\mathbb{Q} [M_T | \mathscr{F}_0] = M_0 = 1$ for any $T$, and by monotonicity we have $0 \leq c(T,k) \leq 1$.
Hence, if we define a set of functions $s(x,y): X \times Y \rightarrow \mathbb{R}$, where $X,Y \subseteq \mathbb{R}_{\geq 0}$ are compact sets, by
\begin{equation}
\begin{aligned}
    \mathcal{S}(X \times Y) = \Bigg\{ & (x,y) \mapsto s(x, y): \forall ~x_1 < x_2 \in X, ~ y_1 < y_2 < y_3 \in Y, \\
     & 0 \leq s \leq 1, ~s(x_1, \cdot) \leq s(x_2, \cdot), \\
     & -1 \leq \frac{s(\cdot,y_2)-s(\cdot,y_1)}{y_2-y_1} \leq \frac{s(\cdot,y_3)-s(\cdot,y_2)}{y_3-y_2} \leq 0 \Bigg\},
\end{aligned}
\label{eq:surface_no_arbitrage}
\end{equation}
then no arbitrage can be constructed on the static surface $(T,k) \mapsto c(T,k)$ if $c \in \mathcal{S}(\mathbb{R}_{> 0} \times \mathbb{R}_{\geq 0})$. Consequently, no static arbitrage can be constructed from the finite collection of prices if 
\begin{equation}
\exists c \in \mathcal{S} \left(\mathcal{T}^e \times [0, \max_{i,j} k^i_j] \right), \text{ s.t. } \forall (T_i, k^i_j) \in \mathcal{P}^{T,k}, ~c(T_i, k^i_j) = c^i_j,
\label{eq:call_shape_constraints}
\end{equation}
where $\mathcal{P}^{T,k} = \{(T_i, k^i_j)\}_{1 \leq j \leq n_i, 1 \leq i \leq m}$.

Condition (\ref{eq:call_shape_constraints}) can be characterised by practically verifiable constraints of prices $\mathbf{c}$. We slightly revise Cousot's construction (Definition 2.1 -- 2.3 in \cite{cousot2004}). We augment the given price data with the price that corresponds to a call struck at $0$ for each expiry. This means $\forall i \in \{1,\cdots,m\}$ we add $K_0^i = 0$ and $C_0^i = F_i$, or equivalently $k^i_0 = 0$ and $c^i_0 = 1$. This augmentation is necessary to check arbitrage relationships between call options and forwards. Define, for any $k^{i_1}_{j_1} > k^{i_2}_{j_2}$, where $1 \leq i_1, i_2 \leq m$, $0 \leq j_1 \leq n_{i_1}$, and $0 \leq j_2 \leq n_{i_2}$,
\begin{equation}
\beta(i_1, j_1; i_2, j_2) := \frac{c^{i_1}_{j_1}-c^{i_2}_{j_2}}{k^{i_1}_{j_1}-k^{i_2}_{j_2}},
\label{eq:beta}
\end{equation}
which can be viewed as the slope of the straight line passing through the two points $(k^{i_1}_{j_1}, c^{i_1}_{j_1})$ and $(k^{i_2}_{j_2}, c^{i_2}_{j_2})$, if we plot all prices on the $(k,c)$ plane. We will employ $\beta(\cdot)$ to define the price of some \textit{test strategies}.

\begin{definition}
A test \textit{spread} strategy is defined $\forall 1 \leq i_1 \leq i_2 \leq m$, and $\forall 0 \leq j_1 \leq n_{i_1}$, $0 \leq j_2 \leq n_{i_2}$ such that $k^{i_1}_{j_1} \geq k^{i_2}_{j_2}$, by 
\begin{equation*}
\text{S}^{i_1, i_2}_{j_1, j_2} = 
\begin{cases}
- \beta(i_1, j_1; i_2, j_2) & $if$ ~k^{i_1}_{j_1} > k^{i_2}_{j_2}, \\
c^{i_2}_{j_2} - c^{i_1}_{j_1} & $if$ ~k^{i_1}_{j_1} = k^{i_2}_{j_2}.
\end{cases}
\end{equation*}
In particular, there are three types of test spread strategies:
\begin{enumerate}[leftmargin=*, label=(\arabic*)]
\item \textit{Vertical spread}: $\text{VS}^i_{j_1, j_2} = \text{S}^{i, i}_{j_1, j_2}$ with $k^{i}_{j_1} > k^{i}_{j_2}$.
\item \textit{Calendar spread}: $\text{CS}^{i_1,i_2}_{j} = \text{S}^{i_1, i_2}_{j_1, j_2}$ with $k^{i_1}_{j_1} = k^{i_2}_{j_2}$ and $i_1 < i_2$.
\item \textit{Calendar vertical spread}: $\text{CVS}^{i_1, i_2}_{j_1, j_2} = \text{S}^{i_1, i_2}_{j_1, j_2}$ with $k^{i_1}_{j_1} > k^{i_2}_{j_2}$ and $i_1 < i_2$.
\end{enumerate}
\end{definition}

\begin{definition}
A test \textit{butterfly} strategy is defined $\forall i, i_1, i_2 \in [1,m]$ s.t. $i \leq i_1$ and $i \leq i_2$, $\forall j \in [0, n_i], ~j_1 \in [0, n_{i_1}], ~j_2 \in [0, n_{i_2}]$ such that $k_{j_1}^{i_1} < k_{j}^{i} < k_{j_2}^{i_2}$, by
\begin{equation*}
\text{B}^{i,i_1,i_2}_{j,j_1,j_2} = - \beta(i,j;i_1,j_1) + \beta(i_2,j_2;i,j).
\end{equation*}
In particular, there are two types of test butterfly strategies:
\begin{enumerate}[leftmargin=*, label=(\arabic*)]
\item \textit{Vertical butterfly}: $\text{VB}^{i}_{j,j_1,j_2} = \text{B}^{i,i,i}_{j,j_1,j_2}$.
\item \textit{Calendar butterfly}: $\text{CB}^{i,i_1,i_2}_{j,j_1,j_2} = \text{B}^{i,i_1,i_2}_{j,j_1,j_2}$ where $i,~i_1$ and $i_2$ are not all equal.
\end{enumerate}
\end{definition}

Based on these definitions of test strategies, we restate Cousot's constraints for no-arbitrage in the following proposition.

\begin{proposition}[Cousot \cite{cousot2007}, \cite{cousot2004}]
\label{prop:cousot}
All test strategies are non-negative, and all test vertical spreads are not greater than $1$, if and only if there exist $m$ risk-neutral measures $\{\mathbb{Q}_{T_i}\}_{1 \leq i \leq m}$ corresponding to all option expiries, that are NDCO. In addition, all their means are equal to $M_0=1$.
\end{proposition}

Together with Kellerer's theorem \cite{kellerer1972}, Proposition \ref{prop:cousot} gives sufficient conditions for the existence of a $\mathbb{Q}$-martingale (thus no static arbitrage), in terms of constraints on prices of the test strategies. Those constraints are also necessary for no static arbitrage if semi-static strategies are allowed to exploit arbitrage opportunities, as proved by Cousot in Appendix A of \cite{cousot2007}.

\subsection{Constraints reduction}
\label{sec:constraints_reduction}

Cousot's constraints contain redundancies. For instance, if two vertical spreads $\text{VS}^i_{j_2,j_1}$ and $\text{VS}^i_{j_3,j_2}$ (where $k^i_{j_1} < k^i_{j_2} < k^i_{j_3}$) are non-negative, then $\text{VS}^i_{j_3,j_1} \geq 0$ holds automatically. We will reduce the number of constraints from $\mathcal{O}(N^3)$ to $\mathcal{O}(m^2 N)$ by localisation on the surface. Localisation can successfully reduce the amount of constraints because the shape constraints specified in (\ref{eq:surface_no_arbitrage}) include only boundedness, positivity, monotonicity and convexity, which are all local properties. The reduced set of constraints is listed in Table \ref{tab:reduced_constraints}, where the (order of the) number of constraints in each category is also indicated.

\begin{table}[!h]
\centering
\footnotesize
\begin{tabular}{llc}
\toprule
\textbf{Category} & \textbf{Constraints} & \textbf{Number} \\ \midrule

\hypertarget{c1}{C1} Outright & $\forall i \in [1,m], ~ c_{n_i}^i \geq 0$ & $m$ \\ \midrule

\multirow{2}{*}{\hypertarget{c2}{C2} Vertical spread} & $\forall i \in [1,m], ~ j \in [1,n_i],$ & \multirow{2}{*}{$N + m$} \\
 & $\text{VS}^i_{j,j-1} \geq 0 ~\text{and } \text{VS}^i_{1,0} \leq 1$ & \\ \midrule

\hypertarget{c3}{C3} Vertical butterfly & $\forall i \in [1,m], ~ j \in [1,n_i-1], ~\text{VB}^i_{j, j-1, j+1} \geq 0$ & $N - m$ \\ \midrule

\multirow{2}{*}{\hypertarget{c4}{C4} Calendar spread} & $\forall 1 \leq i_1 < i_2 \leq m, ~j_1 \in [0,n_{i_1}], ~j_2 \in [0,n_{i_2}],$ & \multirow{2}{*}{$\mathcal{O} (m N)$} \\
 & $\text{CS}^{i_1,i_2}_{j_1,j_2} \geq 0$ & \\ \midrule

\multirow{3}{*}{\begin{tabular}[c]{@{}l@{}}\hypertarget{c5}{C5} Calendar vertical \\ spread\end{tabular}} & $\forall i^* \in [1,m], ~ j^* \in [1,n_{i^*}],$ & \multirow{3}{*}{$\mathcal{O} (m N)$} \\
 & $\text{define } \mathcal{I}:=\{i,j: T_i > T_{i^*}, ~k^{i^*}_{j^*-1} < k^i_j < k^{i^*}_{j^*}\},$ & \\
 & $\text{then } \forall i,j \in \mathcal{I}, ~\text{CVS}^{i^*, i}_{j^*, j} \geq 0$ & \\ \midrule

\multirow{9}{*}{\begin{tabular}[c]{@{}l@{}}\hypertarget{c61}{C6.1} \\ Calendar butterfly I \\ (Absolute location \\ ~convexity)\end{tabular}} & $\forall i^* \in [1,m], ~ j^* \in [1,n_{i^*}-1],$ & \multirow{9}{*}{$\mathcal{O} (m^2 N)$} \\
 & $\text{define } \mathcal{I}:=\{i,j: T_i > T_{i^*}, ~k^{i^*}_{j^*-1} < k^i_j < k^{i^*}_{j^*}\},$ & \\
 & $\text{then } \forall i,j \in \mathcal{I}, ~\text{CB}^{i^*, i, i^*}_{j^*, j, j^*+1} \geq 0;$ & \\ \cmidrule(lr){2-2}
 & $\forall i^* \in [1,m], ~ j^* \in [2,n_{i^*}],$ & \\
 & $\text{define } \mathcal{I}:=\{i,j: T_i > T_{i^*}, ~k^{i^*}_{j^*-1} < k^i_j < k^{i^*}_{j^*}\},$ & \\
 & $\text{then } \forall i,j \in \mathcal{I}, ~\text{CB}^{i^*, i^*, i}_{j^*-1,j^*-2,j} \geq 0;$ & \\ \cmidrule(lr){2-2}
 & $\forall i^* \in [1,m],$ & \\
 & $\text{define } \mathcal{I}:=\{i,j: T_i > T_{i^*}, ~ k^i_j > k^{i^*}_{n_{i^*}} \},$ & \\
 & $\text{ then } \forall i,j \in \mathcal{I}, ~\text{CB}^{i^*, i^*, i}_{n_{i^*}, n_{i^*}-1, j} \geq 0$ & \\ \midrule
 
\multirow{8}{*}{\begin{tabular}[c]{@{}l@{}}\hypertarget{c62}{C6.2} \\ Calendar butterfly II \\ (Relative location \\ ~convexity)\end{tabular}} & $\forall i^* \in [1,m], ~ j^* \in [1,n_{i^*}-1],$ & \multirow{8}{*}{$\mathcal{O} (m^2 N)$} \\
 & $\text{define } \mathcal{I}_1:=\{i,j: T_i > T_{i^*}, ~k^{i^*}_{j^*-1} < k^i_j < k^{i^*}_{j^*}\},$ & \\
 & $\mathcal{I}_2:=\{i,j: T_i > T_{i^*}, ~k^{i^*}_{j^*} < k^i_j < k^{i^*}_{j^*+1}\},$ & \\ 
 & $\forall i_1,j_1 \in \mathcal{I}, ~\forall i_2,j_2 \in \mathcal{I}_2, ~\text{CB}^{i^*, i_1, i_2}_{j^*, j_1, j_2} \geq 0;$ & \\ \cmidrule(lr){2-2}
 & $\forall i^* \in [1,m],$ & \\
 & $\text{define } \mathcal{I}_1:=\{i,j: T_i > T_{i^*}, ~k^{i^*}_{n_{i^*}-1} < k^i_j < k^{i^*}_{n_{i^*}}\},$ & \\ 
 & $\mathcal{I}_2:=\{i,j: T_i > T_{i^*}, ~k^i_j > k^{i^*}_{n_{i^*}}\},$ & \\
 & $\forall i_1,j_1 \in \mathcal{I}, ~\forall i_2,j_2 \in \mathcal{I}_2, ~\text{CB}^{i^*, i_1, i_2}_{n_{i^*}, j_1, j_2} \geq 0$ & \\
\bottomrule
\end{tabular}
\caption{The reduced set of static arbitrage constraints}
\label{tab:reduced_constraints}
\end{table}

We give details of the localisation method in Appendix \ref{appendix:localisationConstraints}. We claim that the reduced set of constraints listed in Table \ref{tab:reduced_constraints} are sufficient to imply Cousot's constraints, thus are sufficient and necessary to guarantee no-arbitrage, as stated in Proposition \ref{prop:reducedConstraint}.

\begin{proposition}\label{prop:reducedConstraint}
If the constraints \hyperlink{c1}{C1} -- \hyperlink{c61}{C6} are satisfied, then all test strategies are non-negative, and all test vertical spreads are not greater than 1.
\end{proposition}
\begin{proof}
See Appendix \ref{appendix:proofReducedConstraint}.
\end{proof}

\section{Arbitrage repair}
\label{sec:arbitrage_repair}

The static arbitrage constraints in Table \ref{tab:reduced_constraints} are linear inequalities of at most three call prices. Therefore, we can write these constraints in the form $A \mathbf{c} \geq \mathbf{b}$, where $\mathbf{c} = [c^1_1 ~\cdots ~c^1_{n_1} ~\cdots ~c^m_{n_m}]^\top \in \mathbb{R}^N$, and $A  = (a_{ij}) \in \mathbb{R}^{R \times N}$ and $\mathbf{b} = (b_j) \in \mathbb{R}^{R}$ are a constant matrix and a vector corresponding to coefficients and bounds of the inequalities, respectively, that are completely determined by the expiries and strikes of observed options. Here, $R$ is the number of no-arbitrage constraints, where $R \sim \mathcal{O}(m^2 N)$. These constraints are \textit{feasible} by construction, i.e. $\{\mathbf{x} \in \mathbb{R}^N: A \mathbf{x} \geq \mathbf{b} \} \neq \emptyset$, because $\mathcal{S}(\mathcal{T}^e \times [0, \max_{i,j} k^i_j] ) \neq \emptyset$, (for example, the prices under a Black--Scholes model satisfy the requirements).

When some row of the system of inequalities $A \mathbf{c} \geq \mathbf{b}$ is not satisfied, there is arbitrage. We define $\bm{\varepsilon}$ to be the vector of perturbations added to the vector of call prices $\mathbf{c}$ such that the perturbed prices are arbitrage-free, i.e. $A (\mathbf{c} + \boldsymbol{\varepsilon}) \geq \mathbf{b}$. Hence, to remove arbitrage from the call price data, we seek the ``minimal'' repair subject to no-arbitrage constraints:

\begin{equation}
\underset{\boldsymbol\varepsilon \in \mathbb{R}^N}{\text{min}} f(\boldsymbol\varepsilon), \quad \text{subject to } A \boldsymbol{\varepsilon} \geq \mathbf{b} - A \mathbf{c},
\label{eq:repair}
\end{equation}
where the objective $f: \mathbb{R}^N \rightarrow \mathbb{R}$ measures how much the perturbation deviates from zero. The formulation (\ref{eq:repair}) is \textit{feasible} because its constraints are feasible.

\subsection{Design of the objective without liquidity consideration}

We start from the simple case where there is no liquidity difference among options. It seems natural to use the $\ell^2$-norm for measuring the size of perturbations due to its convexity and computational efficiency when optimising by gradient-based methods. The $\ell^2$-norm has been widely used in data smoothing algorithms, such as \cite{AitSahalia2003}, \cite{Fengler2009}, \cite{Fengler2012} and \cite{Fengler2015}.

However, the $\ell^2$-norm usually leads to small perturbations for \textit{all} prices, while in our application \textit{sparse} perturbation is desirable. An alternative is the $\ell^0$-norm\footnote{Note that the $\ell^0$-norm is not actually a ``norm'' as it violates the homogeneity and triangle inequality properties that a vector norm must satisfy.}, which is a natural way of comparing difference, and produces sparse solutions. Nevertheless, the $\ell^0$-norm is nonconvex and in general leads to an NP-hard \cite{Natarajan1995} optimisation problem. Hence, it is natural to consider the $\ell^1$-norm, which is well known as a convex relaxation of the $\ell^0$-norm. In fact, optimal solutions of the $\ell^0$ and $\ell^1$ norms objectives are equivalent under certain conditions, see \cite{Candes2005}, \cite{gribonva2003} and \cite{Donoho2003}.

Choosing the $\ell^1$-norm has other benefits. When minimising a convex continuous objective function like the $\ell^1$-norm, every local minimum is a global minimum, see Chapter 4 of \cite{Boyd2004}. In addition, our repair problem is a Linear Programming (LP) problem with the $\ell^1$-norm objective, which can be solved fairly quickly even for large-scale problems. Finally, compared with the $\ell^2$-norm, the $\ell^1$-norm is more robust to outliers because the $\ell^2$-norm squares values, which increases the cost of outliers quadratically, see Huber \cite{Huber1981}.

Consequently, the $\ell^1$-norm is a natural candidate for the objective function.
Blacque-Florentin and Missaoui \cite{Blacque2016} also choose the $\ell^1$-norm as objective when fitting tensor polynomials to sparse data, as inspired by the compressed sensing framework. The differences between our work and theirs are that they are concerned with smoothing data rather than repairing data, and assume a rectangular grid of strikes and expiries. The $\ell^1$-norm optimisation with linear constraints can be expressed as an LP problem. We write the objective function as $f(\bm\varepsilon) := ||\bm\varepsilon||_{\ell^1} = \sum_{i=j}^N |\varepsilon_j| = \sum_{i=j}^N \left( \varepsilon_j^+ + \varepsilon_j^- \right)$, where $\varepsilon_j^+ = \max (\varepsilon_j, 0), \quad \varepsilon_j^- = -\min (\varepsilon_j, 0)$ for each $j$. We denote $\bm\varepsilon^+ = [\varepsilon_1^+ \cdots \varepsilon_N^+]$ and $\bm\varepsilon^- = [\varepsilon_1^- \cdots \varepsilon_N^-]$ so that $\bm\varepsilon = \bm\varepsilon^+ - \bm\varepsilon^-$. We define $B = [-A ~A]$ and $\bm\theta = [\bm\varepsilon^+  \bm\varepsilon^-]^\top$. Hence, the repair problem with the $\ell^1$-norm minimisation is equivalent to the following LP in canonical form:
\begin{equation}
\underset{\bm\theta}{\text{min }} \bm 1^\top \bm\theta, \quad \text{subject to } B \bm\theta \leq  A \mathbf{c} - \mathbf{b}, ~\bm\theta \geq \bm0.
\label{eq:repairLP1}
\end{equation}
After solving for an optimal $\hat{\bm\theta} = [\hat{\bm\varepsilon}^+ ~\hat{\bm\varepsilon}^-]^\top$, the optimal perturbation vector is recovered by $\hat{\bm\varepsilon} = (\hat{\bm\varepsilon}^+ - \hat{\bm\varepsilon}^-)^\top$.

\subsection{Inclusion of bid and ask prices}
\label{sec:include_bid_ask}

The reference prices will typically lie within their corresponding bid-ask price bounds. In the presence of arbitrage, we not only want minimal repair, but also wish to have as many perturbed prices falling within the bid-ask price bounds as possible. Specifically, a reference price with wider bid-ask spread shall be given more freedom to be perturbed. The sparsity of the solution of the $\ell^1$-norm optimisation is less desirable if perturbing a larger number of prices can keep more perturbed prices within the bid-ask price bounds.

\subsubsection*{Design of the objective with bid and ask prices}

We consider using the best bid/ask prices for data repair. To incorporate bid-ask price constraints into the repair problem, we  revise the objective function $f$ rather than adding extra constraints. In other words, we treat bid-ask price bounds as \textit{soft constraints} rather than \textit{hard constraints} like the arbitrage constraints. There may not be arbitrage-free prices within the bid-ask price bounds, and adding bid-ask price bounds as hard constraints may cause the repair problem to be infeasible.

We choose an objective function of the form $f(\bm\varepsilon) = \sum_{j=1}^N f_j(\varepsilon_j)$ with $f_j(x) \geq 0$ for $x\in \mathbb{R}$. Then $f_j(x)$ can be naturally interpreted as the cost of perturbing the $j$-th option price, and $\diff f_j(x) / \diff |x| > 0$ (if defined) gives the marginal cost. The $\ell^1$-norm objective sets $f_j(x) = |x| = \max (-x, x)$ and any perturbation $x$, where $|x|>0$, has marginal cost $1$ for all $j$. Let $\delta^a_j, \delta^b_j > 0$ be ask-reference spread and bid-reference spread for the $j$-th price, respectively. To incorporate these spreads into the objective, we require that $f_j(x)$ should have the following properties, for all $j \in [1,N]$:
\begin{enumerate}[label=(\arabic*)]
\item $f_j(0) = \inf_{x} f_j(x) = 0$. The minimum is attained when there is no perturbation, which is costless to the objective;
\item ${f}_j(x)$ is monotonically increasing (decreasing) for $x>0$ ($x<0$);
\item $f_j(-\delta^b_j) = f_j(\delta^a_j) = \delta_0$, where $\delta_0 \geq 0$ is a constant. The cost of perturbing a price to its bid or ask price is the same for all options;
\item $\diff f_j(x) / \diff |x| = 1$ for $x \in (-\infty, -\delta^b_j) \cup (\delta^a_j, +\infty)$. The marginal cost of perturbing a price out of the bid-ask price bounds is the same for all options.
\end{enumerate}

We therefore propose the following objective that meets all the properties, and, with particular merit, retains the ability to be expressed as an LP:
\begin{equation*}
f_j (x) = \max \left( -x-\delta^b_j+\delta_0, ~-\frac{\delta_0}{\delta^b_j} x, ~\frac{\delta_0}{\delta^a_j} x, ~x - \delta^a_j + \delta_0 \right),
\end{equation*}
with $\delta_0 \leq \min(\delta^a_j, \delta^b_j)$ for all $j \in [1,N]$, as such the marginal cost of perturbing a price within the bid-ask price band is not greater than the marginal cost of perturbing mid prices outside the bid-ask price bounds.
\begin{figure}[!h]  
\centering 
  \begin{subfigure}[b]{0.31\linewidth}
    \begin{tikzpicture}[scale=1][
   	thick,
    >=stealth']
  	
  	\coordinate (O) at (0,0);

 	\draw[->] (-1.8,0) -- (1.8,0) coordinate[label = {below:$x$}] ();
  	\draw[->] (0,-0.1) -- (0,2.5) coordinate[label = {right:$f_j^{\ell^1}(x)$}] (ymax);
	
	\draw[gray] (0,0)--(-1.5,1.5);
	\draw[gray] (0,0)--(1.5, 1.5);	
	
\end{tikzpicture}
    \caption{$f_j^{\ell^1}(x)$} \label{fig:repairObj1}  
  \end{subfigure}
  \begin{subfigure}[b]{0.66\linewidth}
	\begin{tikzpicture}[scale=1][
   	thick,
    >=stealth']
  	
  	\coordinate (O) at (0,0);

 	\draw[->] (-3.8,0) -- (3.8,0) coordinate[label = {below:$x$}] ();
  	\draw[->] (0,-0.1) -- (0,2.5) coordinate[label = {right:${f}_j(x)$}] (ymax);
	
	\draw[dashed] (-3.5,.8) -- (3.5,.8); 	
	
	\draw[gray] (0,0) -- (-1.2, .8);
	\draw[gray] (-1.2,.8) -- (-2.4, 2);
	\draw[gray] (0,0) -- (2, 2) node[above, black]{${f}_1(x)$};
	
	\draw[gray] (0,0) -- (-1.8,.8);
	\draw[gray] (0,0) -- (1.8, .8);	
	\draw[gray] (-1.8,.8) -- (-3,2);
	\draw[gray] (1.8, .8) -- (3, 2) node[above, black]{${f}_2(x)$};
	
	\draw[gray] (0,0) -- (-2.8,.8);
	\draw[gray] (0,0) -- (2.6, .8);	
	\draw[gray] (-2.8,.8) -- (-3.3,1.3);
	\draw[gray] (2.6, .8) -- (3.1, 1.3) node[above, black]{${f}_3(x)$};
	
	\draw[<->] (-3.2,0) -- node[fill=white, minimum height=0.2cm, text height=0.2cm] {$\delta_0$} (-3.2,.8);
	
	\draw[dashed] (-1.2,.8) -- (-1.2,0) node[below]{$-\delta^b_2$};
	\draw[dashed] (-1.8,.8) -- (-1.8,0) node[below]{$-\delta^b_1$};
	\draw[dashed] (-2.8,.8) -- (-2.8,0) node[below]{$-\delta^b_3$};
	\draw[dashed] (.8,.8) -- (.8,0) node[below]{$\delta^a_1$};
	\draw[dashed] (1.8,.8) -- (1.8,0) node[below]{$\delta^a_2$};
	\draw[dashed] (2.6,.8) -- (2.6,0) node[below]{$\delta^a_3$};
	
\end{tikzpicture}
	\caption{${f}_j(x)$} \label{fig:repairObj2}
  \end{subfigure}
\caption{Plot of the objective function component $f_j^{\ell^1}(x)$ and ${f}_j(x)$}
\label{fig:repairObj}
\end{figure}
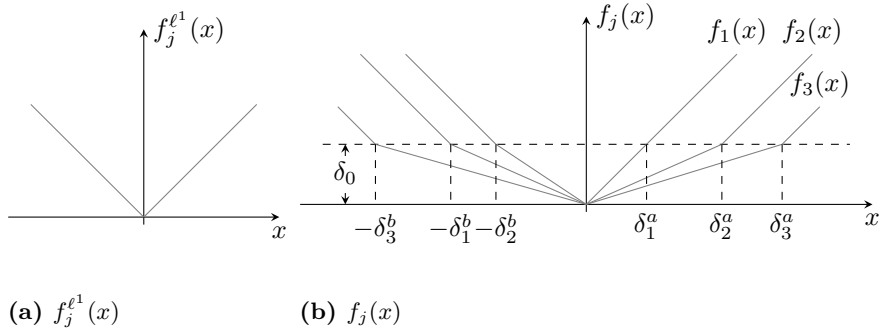

Denote $f_j^{\ell^1}$ as the $j$-th component of the $\ell^1$-norm objective. We visualise the difference between $f_j^{\ell^1}$ and ${f}_j$ in Figure \ref{fig:repairObj}. Note that $f_j^{\ell^1}$ is a special case of ${f}_j$ when $\delta^a_j = \delta^b_j = \delta_0 > 0$ for all $j$. Choosing smaller $\delta_0$ makes it relatively more costly to move prices outside of their bid-ask price bounds. Nevertheless, letting $\delta_0 = 0$ causes the optimisation problem to be ill-posed as it admits infinitely many solutions. For example, if $\varepsilon_j^*=0$ is optimal, then so is $\varepsilon_j^* = \omega \min(\delta^a_j,\delta^b_j)$ for all $\omega \in [0,1]$. In practice, we will choose
\begin{equation}
\delta_0 = \frac{1}{N} \wedge \min_{j=1,\dots,N} \left( \delta^a_j \wedge \delta^b_j \right).
\label{eq:delta_0}
\end{equation}
This means we prefer to move all options (by $\bm\varepsilon$) within the bid-ask, rather than moving one option outside its bid-ask bounds.

Hence, the objective function taking into account bid-ask spread is
\begin{equation}
f(\bm\varepsilon) = \sum_{j=1}^N \max \left( -\mathbf{e}_j^\top \bm\varepsilon -\delta^b_j+\delta_0, ~-\frac{\delta_0}{\delta^b_j} \mathbf{e}_j^\top \bm\varepsilon, ~\frac{\delta_0}{\delta^a_j} \mathbf{e}_j^\top \bm\varepsilon, ~\mathbf{e}_j^\top \bm\varepsilon - \delta^a_j + \delta_0 \right),
\label{eq:repair_obj_ba}
\end{equation}
where $\mathbf{e}_j$ is the standard basis vector for $\mathbb{R}^N$ with its $j$-th element being 1 and others being 0. With objective (\ref{eq:repair_obj_ba}), we can rewrite the repair problem (\ref{eq:repair}) as the following LP by introducing auxiliary variables $\mathbf{t} = [t_1 ~\cdots ~t_N]^\top$:
\begin{equation}
\begin{aligned}
& \underset{\bm\varepsilon, \mathbf{t}}{\text{minimise}}
& & \sum_{j=1}^{N} t_j \\
& \text{subject to}
& & -\varepsilon_j - \delta^b_j + \delta_0 \leq t_j, ~\varepsilon_j - \delta^a_j + \delta_0 \leq t_j, \quad \forall j \in [1,N], \\
& & & -\frac{\delta_0}{\delta^b_j} \varepsilon_j \leq t_j, ~\frac{\delta_0}{\delta^a_j} \varepsilon_j \leq t_j, \quad \forall j \in [1,N], \\
& & & - A \bm\varepsilon \leq  - \mathbf{b} + A \mathbf{c}.
\end{aligned}
\label{eq:repair_PLO}
\end{equation}

After solving for the optimal perturbation vector $\hat{\bm\varepsilon}$, we get the arbitrage-free normalised call price $\hat{\mathbf{c}} = \mathbf{c} + \hat{\bm\varepsilon}$. For each $i,j$, the arbitrage-free call price is $\widehat{C}^i_j = \hat{c}^i_j D_i F_i$.

\subsubsection*{Executable arbitrage opportunities}
\label{sec:exe_arbitrage}

We refer to the objective function taking into account bid-ask spread with $\delta_0$ as in (\ref{eq:delta_0}) as the $\ell^1$-BA objective. We define the \textit{effectively perturbed prices} as those that are perturbed outside of the bid-ask price bounds. We denote the number of perturbed (resp. effectively perturbed) prices by $N^\varepsilon$ (resp. $N^{\varepsilon, \delta}$), thus
\begin{equation}
N^{\varepsilon} = \sum_{j=1}^N \mathbbm{1}_{\{|\varepsilon_j| > 0\}}, \quad N^{\varepsilon, \delta} = \sum_{j=1}^N \mathbbm{1}_{\{\varepsilon_j > \delta_j^a\} \cup \{\varepsilon_j < -\delta_j^b\}}.
\label{eq:eff_pertub}
\end{equation}

We say an arbitrage is \textit{executable} if we can realise it by buying and selling its components at their ask and bid quotes, respectively. The arbitrage detected in options' reference prices is not necessarily executable. However, if the $\ell^1$-BA repair results in effective perturbations, i.e. $N^{\varepsilon, \delta} > 0$, then there must exist executable arbitrages. To see this, let $E_j = [c_j - \delta^b_j, c_j + \delta^a_j]$, and we can characterise $N^{\varepsilon, \delta} > 0$ as
\begin{equation*}
    \text{if } \forall i \in [1,R], ~\sum_{j=1}^N a_{ij} \hat{c}_j \geq b_i,
    \text{ then } \exists j \in [1,N] ~\text{s.t. } \hat{c}_j \not\in E_j. 
\end{equation*}
Equivalently, its contrapositive statement is
\begin{equation*}
    \text{if } \forall j \in [1,N], ~\hat{c}_j \in E_j,
    \text{ then } \exists i^* \in [1,R] ~\text{s.t. } \sum_{j=1}^N a_{i^*j} \hat{c}_j < b_{i^*}.
\end{equation*}
Therefore, it holds that
\begin{equation*}
    \sum_{j=1}^N a_{i^*j} \left[ (c_j + \delta^a_j) \mathbbm{1}_{\{a_{i^*j} \geq 0\}} + (c_j - \delta^b_j) \mathbbm{1}_{\{a_{i^*j} < 0\}} \right] < b_{i^*}.
\end{equation*}
By going long on the left-hand side and going short on the right side of the inequality, we construct a portfolio that makes immediate positive profit, while the portfolio has non-negative future payoffs. The left-hand side of the inequality consists of positions in options, for which we buy at ask price $(c_j + \delta^a_j)$ and sell at bid price $(c_j - \delta^b_j)$.

\section{Empirical studies}
\label{sec:numerics}

We carry out a series of empirical studies. We show that arbitrage is frequently present in historical price data, so repairing data is important. We also demonstrate empirical performances of the repair method in terms of sparsity, speed and improvement to model calibration. Last, we use the $\ell^1$-BA repair for identifying the formation and disappearance of arbitrage in the intra-day S\&P 500 options market on a day when the market underwent a regime switch.

\subsection{Presence of arbitrage in historical option price data}

We collect daily close (bid, ask and mid) prices from 1st November, 2007 to 31st May, 2018 for OTC FX options from Bloomberg. Bloomberg provides price quoted as implied volatility given in terms of delta. We choose 13 benchmark tenors (expiries) from overnight (one-day) to two-year. For each tenor, a list of standard \textit{instruments} are available: at-the-money (ATM), risk-reversal (RR) and butterfly (BF). We choose the liquid 10-delta, 15-delta, 25-delta and 35-delta instruments, and construct a vanilla volatility smile of 9 moneynesses for each tenor. Following the OTC FX market conventions \cite{Wystup2017}, we compute strike and time-to-expiry for each IV mid quote, and generate vanilla IV spreads from the bid/ask quotes for the instruments\footnote{Given the \textit{instrument bid-ask spreads} for ATM, RR and BF, one cannot uniquely determine the corresponding \textit{vanilla spreads} without specifying some rule. For example, in practice, trading desks may estimate vanilla spreads only using ATM spreads, which makes the spread of each option at the same expiry equal, see Section 4.2.1 of \cite{Wystup2017}. Since vanilla IVs are linear transformations of instrument IVs, we conservatively assume that vanilla spreads are weighted sums of instrument spreads. This does not take into account that delta-symmetric vanilla spreads are dependent on each other, and generates the widest possible bid-ask spreads for vanilla IVs.}. Thereafter, we calculate mid call prices and vanilla call price spreads using the mid vanilla IVs and generated vanilla IV spreads, together with Bloomberg FX mid forward curves. There are $117 = 13 \text{ (tenors)} \times 9 \text{ (deltas and ATM)}$ data points on the call price surface for each day. In Figure \ref{fig:sampleDataOTC}, we show an example of OTC call option mid-prices and IV mid-quotes on one day.
\begin{figure}[!h]
\centering
\includegraphics[scale=.59]{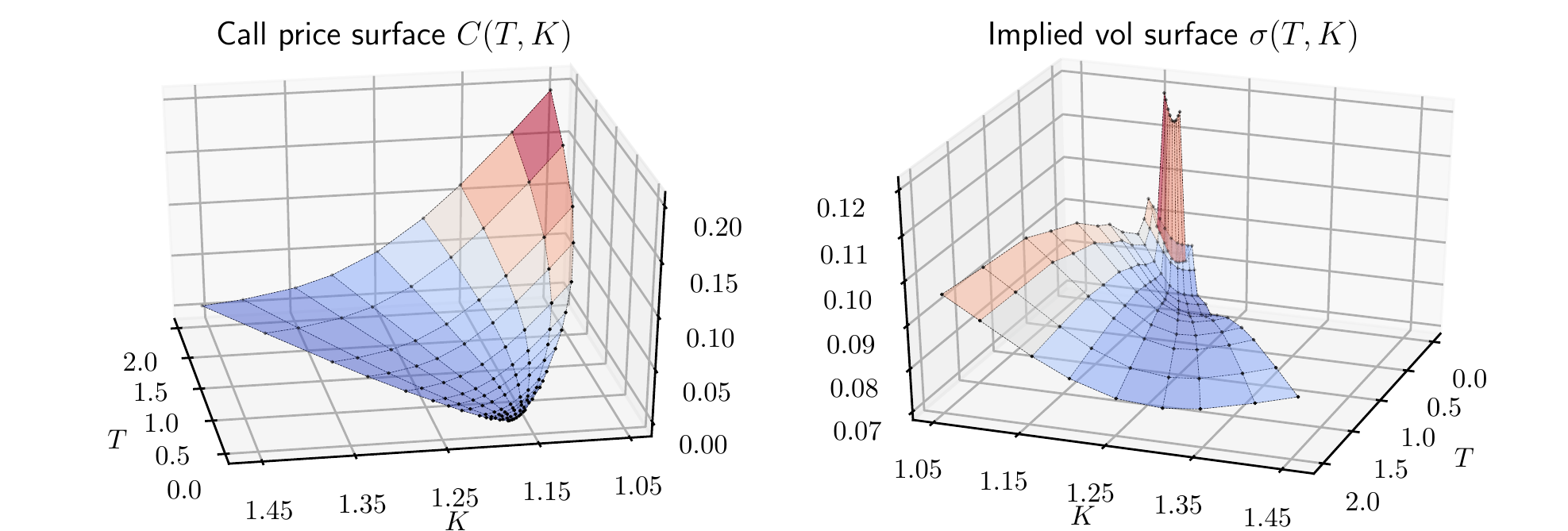}
\vspace{-4mm}
\caption{An example of observed OTC-traded call option prices. These are end of day prices settled by Bloomberg for EURUSD European call options as of 31st May, 2018.}
\label{fig:sampleDataOTC}
\end{figure}

We count violations of arbitrage constraints in raw daily close mid-prices over time for some major currencies and emerging market (EM) currencies. In Figure \ref{fig:detectionTsOtc}, we see that there are more arbitrages in the EM currency markets. We also see persistent clustering of (mild) arbitrages from early 2007 to mid 2012 in major currency markets. Further investigation suggests that these are caused by over-priced 1-day options, which result in calendar arbitrages with longer-dated options. We conjecture that the systematic appearance of the same type of arbitrage is due to Bloomberg's legacy data cleansing method. 

\begin{figure}[!h]
\centering
\includegraphics[scale=0.59]{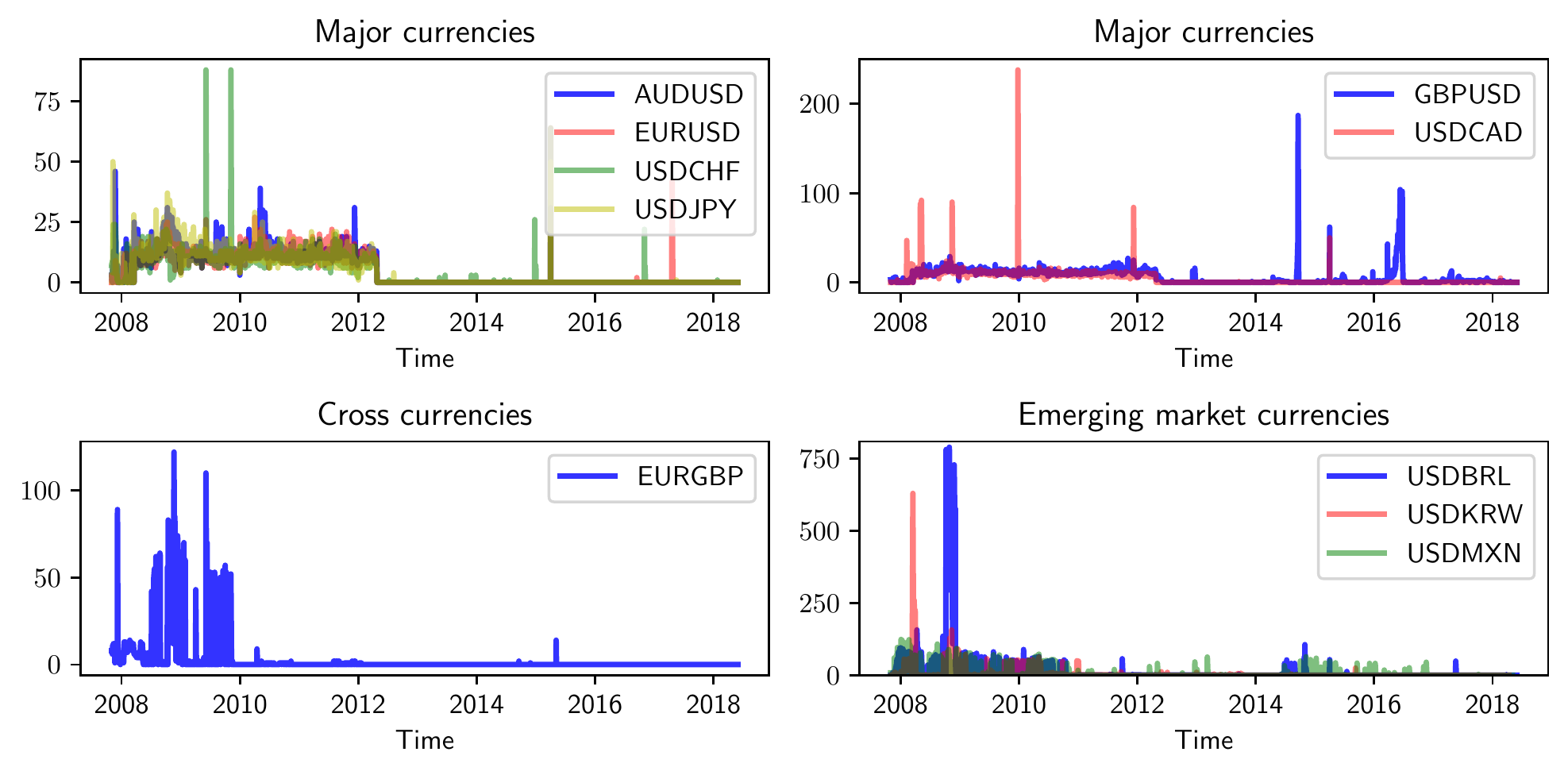}
\vspace{-4mm}
\caption{Time series of number of daily violated arbitrage constraints in OTC FX option market, during the period from 1st November, 2007 to 31st May, 2018.}
\label{fig:detectionTsOtc}
\end{figure}

Calendar arbitrage (especially CVS \hyperlink{c5}{C5} and CBS \hyperlink{c61}{C6}) is more difficult and costly to exploit than non-calendar arbitrages, as it requires rebalancing the hedging portfolio over time. Most arbitrage-free smoothing algorithms in the literature only remove calendar arbitrage of \hyperlink{c4}{C4} type, because they assume a rectangular grid of expiries and strikes. However, calendar arbitrage can be a major source of arbitrage. In Figure \ref{fig:detectionCalendarRatio}, we consider what fraction of the arbitrages are of calendar type for different currency pairs. Comparing medians (and overall distributions), as shown in the plot, the proportion of calendar arbitrages for major currencies (AUD, EUR, GBP, CAD, CHF, and JPY) is larger than that for EM currencies (BRL, KRW, and MXN), though the cross pair EURGBP is an exception. In fact, the medians are very close to 100\% for almost all major currencies except sterling. In other words, nearly all arbitrages in major currencies' option markets are calendar ones.

\begin{figure}[!h]
\centering
\vspace{-4mm}
\includegraphics[scale=0.59]{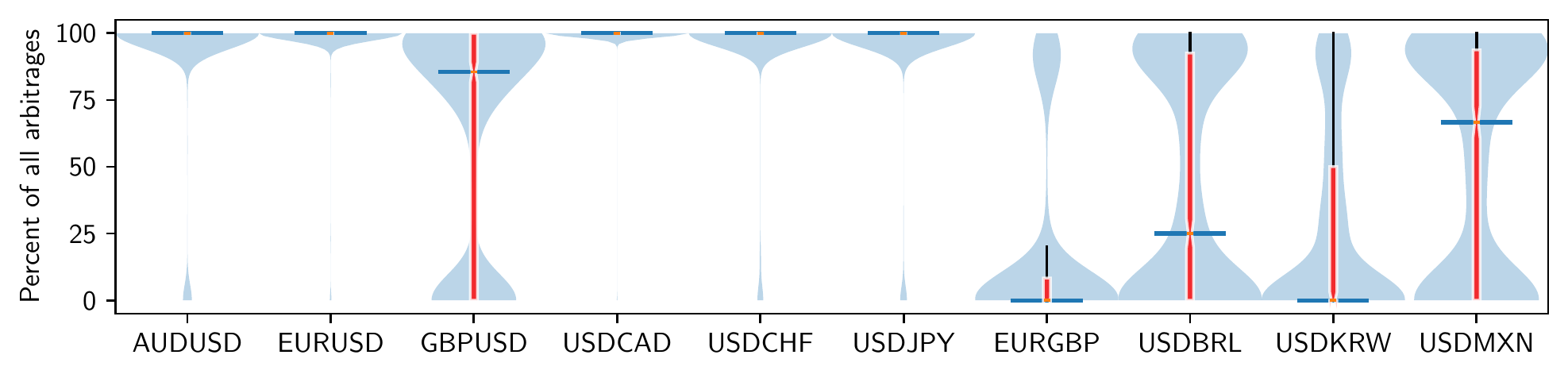}
\vspace{-4mm}
\caption{Fraction of calendar arbitrages on a given day, for different currency pairs during the period from 1st November, 2007 to 31st May, 2018. The light blue shadow is a violin plot which indicates the kernel density of the percentages, and the red notched box is a box plot. The horizontal short bar shows the median of each sample.}
\label{fig:detectionCalendarRatio}
\end{figure}

We examine the day when the EURUSD option price data have the most occurrences of calendar arbitrages over our observation period, and plot the call price curves for the first three expiries in Figure \ref{fig:repairCalendarEURUSD}. There is no non-calendar arbitrage on that day since each curve is non-increasing and convex. After the repair, the $T_1$-curve is pushed downwards until it does not lie beyond the other two curves, which ensures NDCO marginal risk-neutral measures.
\begin{figure}[!h]
\centering
\vspace{-4mm}
\includegraphics[scale=0.59]{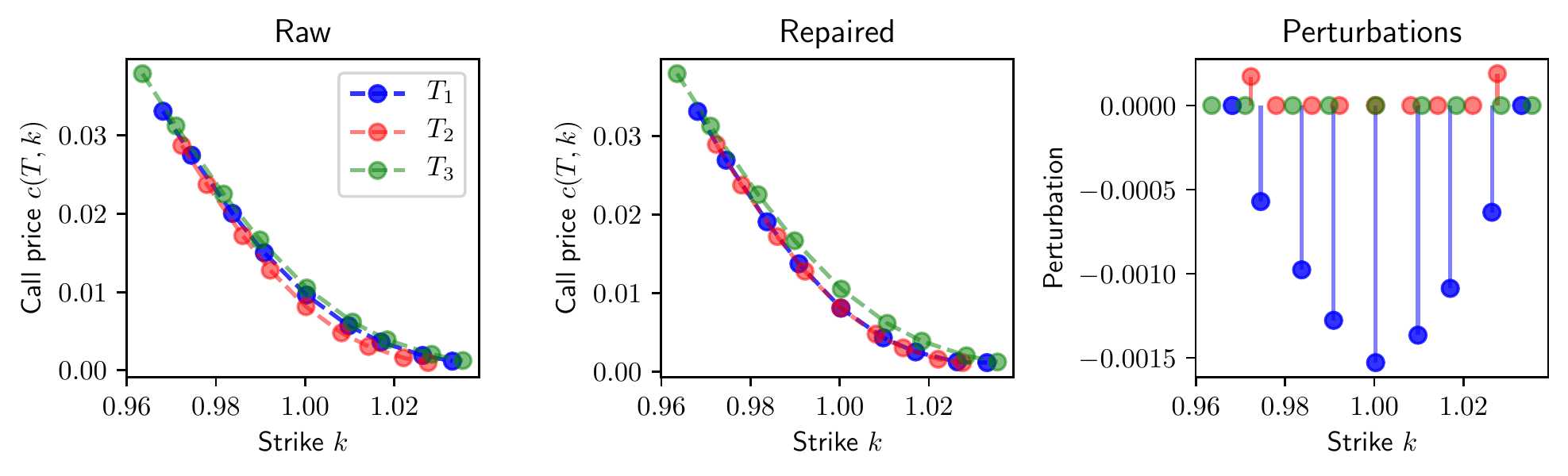}
\vspace{-4mm}
\caption{An example of arbitrage repair for EURUSD call options on 2nd April, 2015. \textit{Left} -- raw call price curves for the first three expiries. \textit{Middle} -- repaired arbitrage-free call price curves. \textit{Right} -- perturbations added to each data point.}
\label{fig:repairCalendarEURUSD}
\end{figure}

However, when calendar and non-calendar arbitrages are mixed, the perturbations added to ensure no arbitrage tend to be more varied in signs. For instance, in Figure \ref{fig:repairCalendarUSDBRL} we plot the call price curves for the first four expiries on the day when USDBRL options had the most occurrences of calendar arbitrage, however, there are also many non-calendar arbitrages. Unlike the above EURUSD example, the repair does not simply translate any curve. Therefore, the perturbations are not systematically negative. 
\begin{figure}[!h]
\centering
\vspace{-4mm}
\includegraphics[scale=0.59]{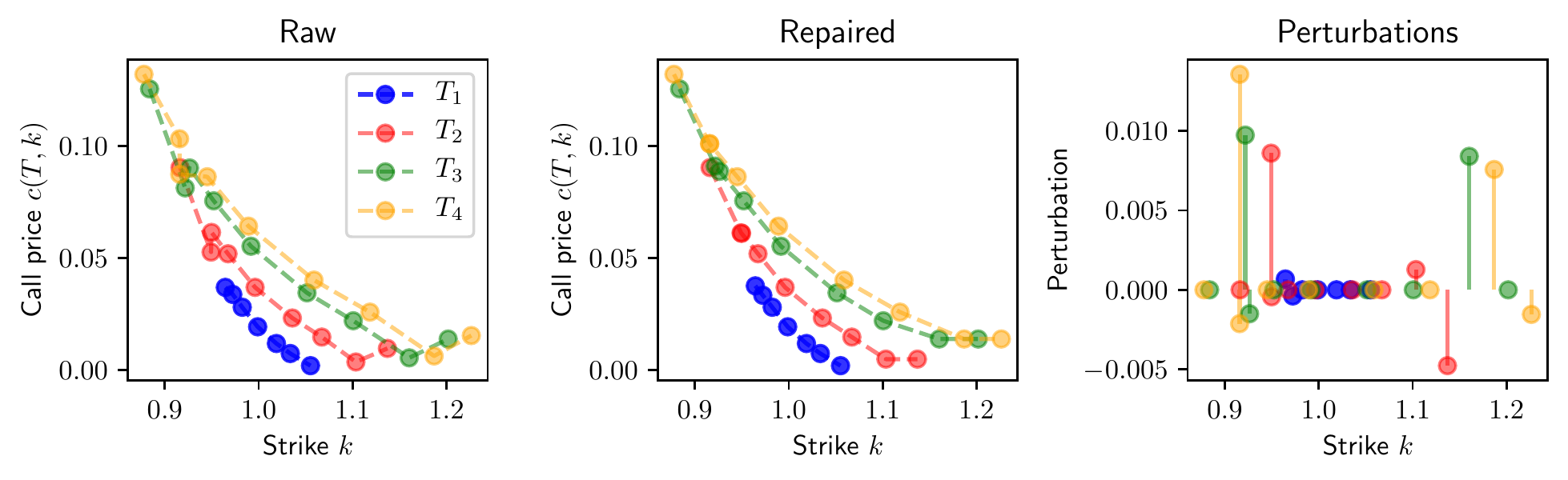}
\caption{An example of arbitrage repair for USDBRL call options on 28th October, 2008.}
\label{fig:repairCalendarUSDBRL}
\end{figure}

\vspace{-4mm}
\subsection{Properties of the repair method}

\subsubsection*{Sparse solution of the $\ell^1$-norm objective}

The $\ell^1$-norm objective leads to sparse perturbations. We show the fraction of perturbed prices in Figure \ref{fig:repairPerturbNumber}. The medians are very close to zero for all currency pairs, indicating that very few data points need to be perturbed on average to remove arbitrage. This is especially true for major currencies, as their distributions collapse almost entirely to zero.

\begin{figure}[!h]
\centering
\vspace{-4mm}
\includegraphics[scale=0.59]{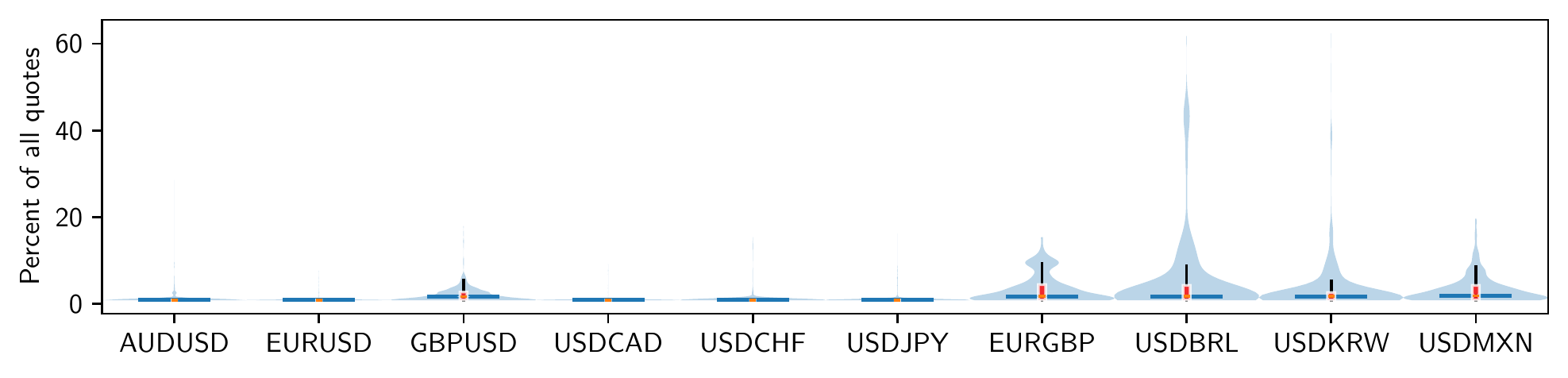}
\caption{Number of perturbed prices as a percentage of all prices, for different currency pairs during the period from 1st November, 2007 to 31st May, 2018.}
\label{fig:repairPerturbNumber}
\end{figure}



\subsubsection*{Computational time}

Our data repair method is designed to be fast due to the LP formulation. In addition, the reduction of arbitrage constraints shrinks the scale of the LP and speeds up the repair. We investigate the computational time of our repair method when applied to a few practical cases. All of the following studies were carried out on a quadcore Intel Core i7-8650U CPU with 32GB RAM. All LPs are solved using the GLPK (GNU Linear Programming Kit) solver wrapped by the CVXOPT \cite{cvxopt} Python package. 

In Figure \ref{fig:timeStats_otc}, we plot histograms of (1) the number of constraints $R \sim \mathcal{O}(m^2 N)$, (2) the fraction of violated constraints, and (3) the elapsed times for constructing the constraints (Table \ref{tab:reduced_constraints}) and solving the LP (\ref{eq:repairLP1}). We take EURUSD and USDBRL as representatives of major currencies and EM currencies, respectively. The number of constraints rarely exceeds 4000, while it takes less than 0.4 seconds to transform them to the matrix form in most cases. Solving the LP with $N=117$ variables and $R<4000$ constraints takes less than 0.05 seconds for EURUSD options and less than 0.1 seconds for USDBRL options. More violated constraints lead to higher computational time.

\begin{figure}[!h]
\centering
\vspace{-4mm}
\begin{subfigure}[b]{\textwidth}
    \includegraphics[scale=0.59]{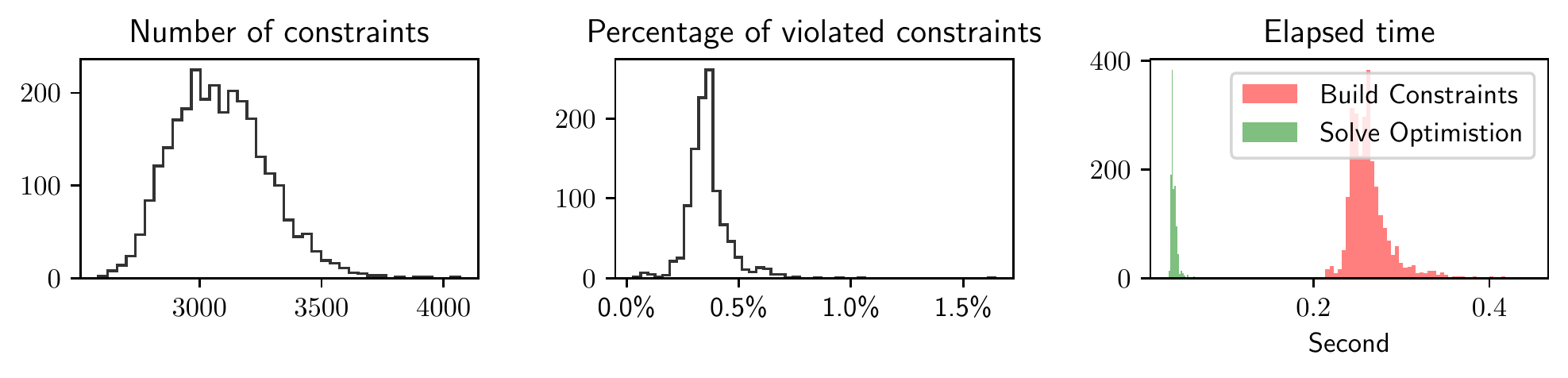}
    \caption{EURUSD options}
\end{subfigure} \\
\begin{subfigure}[b]{\textwidth}
    \includegraphics[scale=0.59]{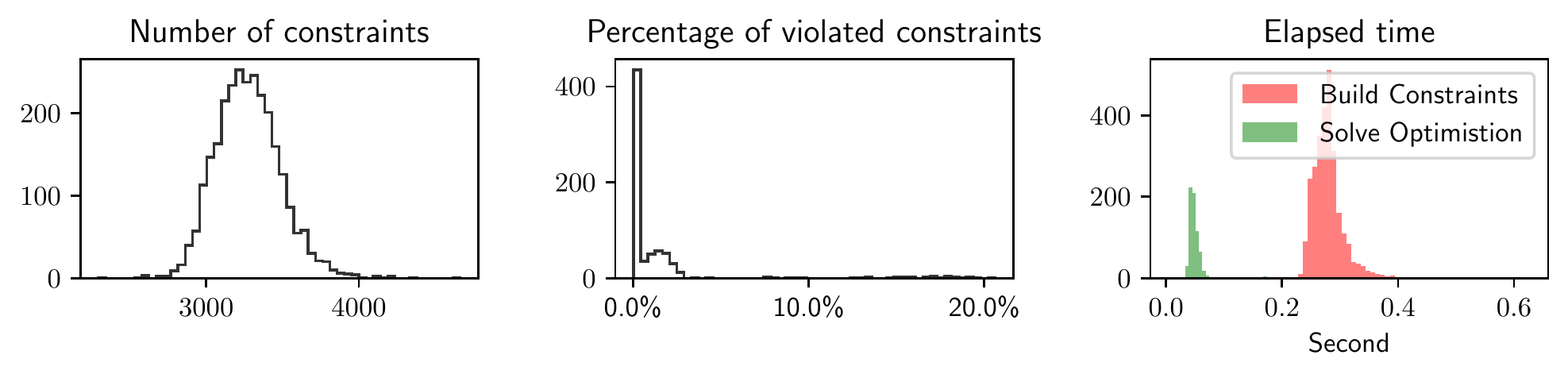}
    \caption{USDBRL options}
\end{subfigure} \\
\caption{Histograms of various statistics for repairing data of EURUSD options and USDBRL options, during the period from 1st November, 2007 to 31st May, 2018.}
\label{fig:timeStats_otc}
\end{figure}

To see how the repair method performs for larger-scale problems, we collect daily settled prices for all \textit{traded} EURUSD call options listed by CME from 1st January, 2013 to 31st December, 2018. The number of traded options varies from one day to another, see Figure \ref{fig:num_tradedoptions}. We show the distribution of traded expiry and strikes on a typical day in Figure \ref{fig:sampleDataETD}. In Figure \ref{fig:timeStats_etd_EURUSD}, we plot similar repair statistics to those for the OTC data. On average, there are 500 call prices per day, which result in on average 25000 arbitrage constraints to verify. Though the number of constraints is observed as high as 90000, it takes less than 1 second to construct them in the matrix form. Solving the LP now can take up to 6 seconds, but on average it only takes 1.44 seconds.

\begin{figure}[!h]
\centering
\begin{subfigure}[b]{0.4\textwidth}
    \centering
    \includegraphics[scale=0.59]{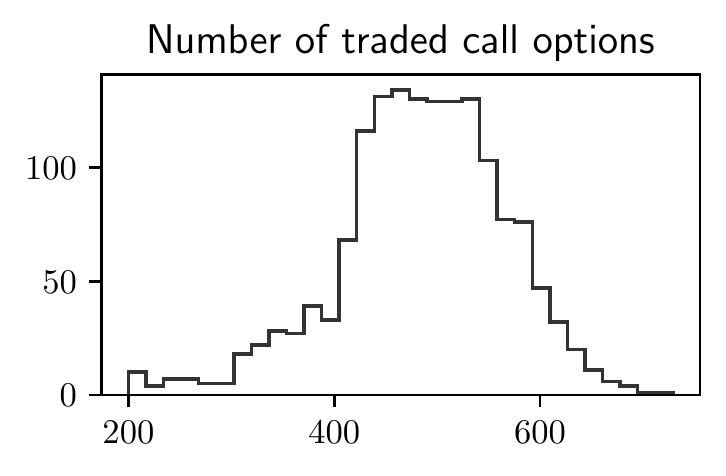}
    \caption{Histogram of the number of traded call options per day.}
    \label{fig:num_tradedoptions}
\end{subfigure}
\begin{subfigure}[b]{0.56\textwidth}
    \centering
    \includegraphics[scale=0.59]{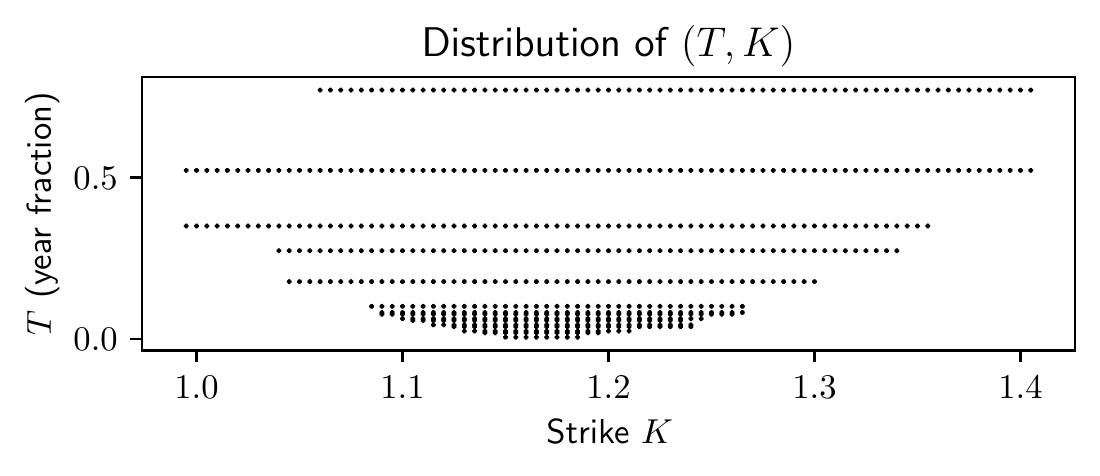}
    \caption{The distribution of $(T,K)$ for all traded EURUSD call options on 31st May, 2018.}
    \label{fig:sampleDataETD}
\end{subfigure}
\\
\begin{subfigure}[b]{\textwidth}
    \includegraphics[scale=0.59]{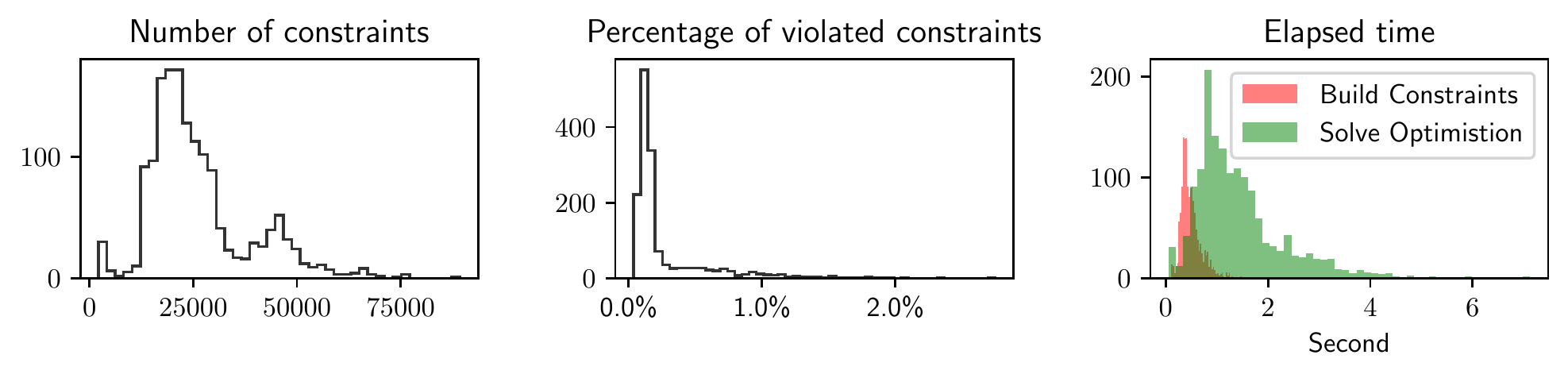}
    \caption{Histograms of various statistics.}
    \label{fig:timeStats_etd_EURUSD}
\end{subfigure} \\
\caption{Statistics for repairing data of CME-listed EURUSD options during the period from 1st January, 2013 to 31st December, 2018.}
\label{fig:timeStats_etd}
\end{figure}

\subsubsection*{Stress testing the $\ell^1$-norm objective repair}
\label{sec:stress_testing}

We test how our repair method works in hypothetical extreme scenarios when there is massive arbitrage. First, we collect arbitrage-free call prices for a day, denote these data by $\mathbf{c} \in \mathbb{R}^N$, and let $\mathcal{I} = \{1,\dots, N\}$ be the set of its indices. Next, we simulate noises and add them to a portion $\lambda \in (0,1]$ of the price data, where we denote $\mathcal{I}_\xi \subset \mathcal{I}$ as the set of indices of those polluted prices. Here, $\mathcal{I}_\xi$ is randomly sampled without replacement such that $|\mathcal{I}_\xi| = \lceil \lambda | \mathcal{I} | \rceil$. Constructing the noise $\bm\xi = (\xi_j) \in \mathbb{R}^N$ by taking  $\xi_j = \zeta_j \mathbbm{1}_{\{j \in \mathcal{I}_\xi\}}$ where $\bm\zeta = (\zeta_j)$ are i.i.d., we then define the noisy price $\tilde{\mathbf{c}} \in \mathbb{R}^N$ by
\begin{equation*}
\tilde{c}_j = c_j e^{\xi_j}, ~\forall 1 \leq j \leq N.
\end{equation*}
The noisy price vector contains arbitrage in general, which we repair by seeking an optimal perturbation $\bm\varepsilon \in \mathbb{R}^N$. The perturbed arbitrage-free price vector is $\hat{\mathbf{c}} = \tilde{\mathbf{c}} + \bm\varepsilon$. An example of $\mathbf{c}$, $\tilde{\mathbf{c}}$, and $\hat{\mathbf{c}}$ is given in Figure \ref{fig:stressTesting}.

\begin{figure}[!h]
\centering
\includegraphics[scale=0.59, trim={1.7cm 1.cm 1.7cm 0cm}, clip]{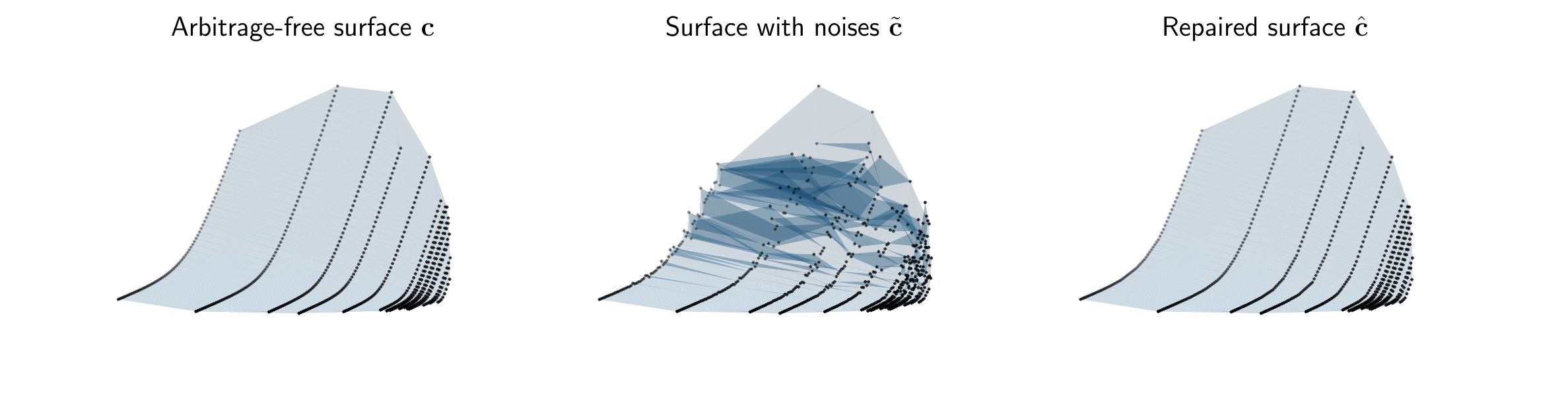}
\caption{An example used for stress testing the repair method. There are $N = 591$ prices in total. The data used are CME-traded EURUSD options' prices on 31st May, 2018.}
\vspace{-2mm}
\label{fig:stressTesting}
\end{figure}

We assess how well $\mathbf{c}$ is recovered by examining the quantities $\ln (\hat{c}_j / c_j)$ and $\hat{\lambda} := \frac{1}{N}\sum_{j=1}^N \mathbbm{1}_{\{\hat{c}_j \neq c_j\}}$. Note that $\ln (\hat{c}_j / c_j) \approx (\hat{c}_j - c_j) / c_j$ (when $\hat{c}_j \approx c_j$). In addition, $\hat{\lambda}$ counts the portion of different prices per call surface. It is unrealistic to expect any repair method to fully recover a price as it is unlikely to know the exact marginal that generates the price. However, given the ground truth that a portion $\lambda$ of the surface prices has been polluted by noise, a desirable data repair method should leave as many unpolluted prices unchanged as possible, i.e. $\hat{\lambda} - \lambda$ should be small.

Assuming Gaussian noises $\bm\zeta \sim \mathcal{N}(\mathbf{0}, \sigma_\xi I)$, we simulate noises $M$ times, and compute the average value of $\hat{\lambda}$ and plot the histograms of $\ln (\hat{c}_j / c_j)$, conditional on non-zero values, as shown in Figure \ref{fig:stressTesting2}. For a fixed noise magnitude $\sigma_\xi$, the gap between $\hat{\lambda}$ and $\lambda$ widens as $\lambda$ increases, i.e. the repair method adjusts a larger number of prices to remove arbitrage if there are more noisy prices. The same observation holds for different values of $\sigma_\xi$, though larger noise magnitude $\sigma_\xi$ results in more arbitrages. Note that taking $\sigma_\xi=1$ and $\lambda=25\%$ already results in, on average, $\hat{\lambda} = 30.80\%$ of the price data being perturbed, an extremely large fraction that has rarely been seen in our data, see Figure \ref{fig:repairPerturbNumber}. Hence, in practice our repair method seems to only perturb a few additional (i.e. $\hat{\lambda} - \lambda \approx 5\%$ for $\lambda = 25\%$) prices to ensure no arbitrage.

\begin{figure}[!h]
\centering
\begin{subfigure}[b]{\textwidth}
    \includegraphics[scale=0.59]{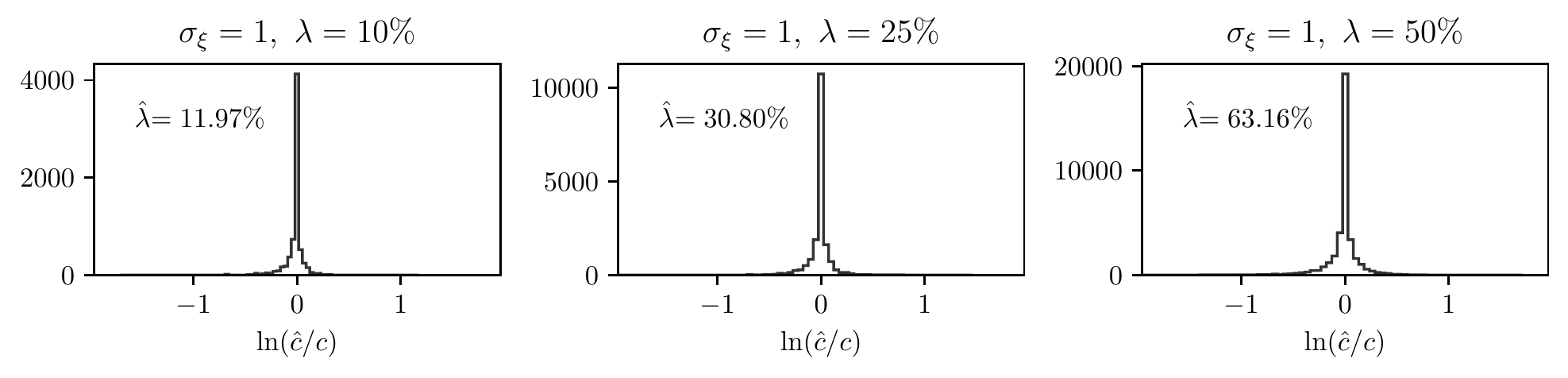}
    \caption{$\sigma_\xi=1$.}
\end{subfigure} \\
\begin{subfigure}[b]{\textwidth}
    \includegraphics[scale=0.59]{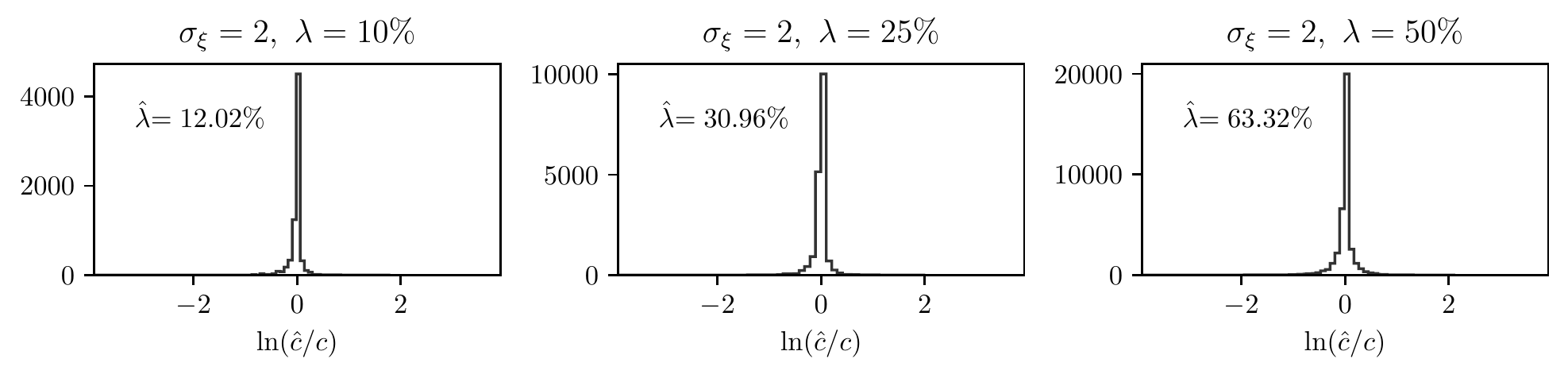}
    \caption{$\sigma_\xi=2$.}
\end{subfigure} \\
\caption{Histograms of $\ln(\hat{c}_j / c_j) \approx (\hat{c}_j - c_j) / c_j$, conditional on non-zero values, computed under differently valued noise simulation parameters $(\lambda, \sigma_\xi)$. We simulate $M=100$ times.}
\label{fig:stressTesting2}
\end{figure}

\subsubsection*{Comparing the objectives: $\ell^1$-norm and $\ell^1$-BA}

The $\ell^1$-BA repair is designed to perturb more prices (larger $N^\varepsilon$) than the $\ell^1$-norm repair does, but fewer of them are effective (smaller $N^{\varepsilon, \delta}$), if possible. To verify this, we apply the $\ell^1$-BA repair method to the same OTC FX option price data. In Figure \ref{fig:statsL1L1BA}, we show the histograms of the difference in these two statistics $N^\varepsilon$ and $N^{\varepsilon, \delta}$ that are produced by the two repair methods.

\begin{figure}[!h]
\centering
\vspace{-4mm}
\includegraphics[scale=0.59]{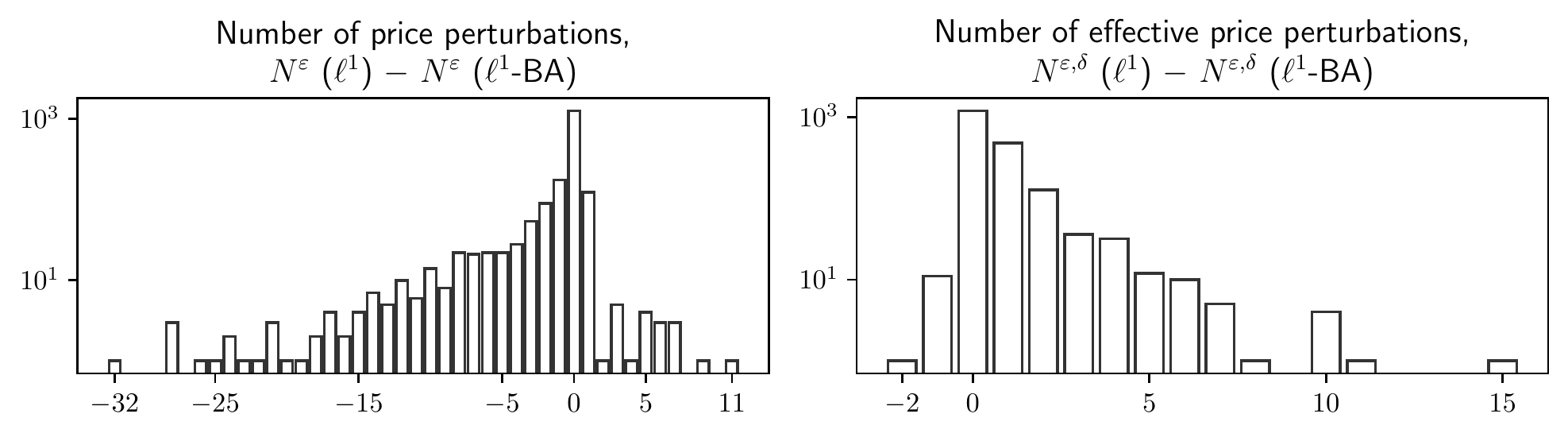}
\caption{Histograms of the difference in $N^\varepsilon$ and $N^{\varepsilon, \delta}$ that are produced by the $\ell^1$-norm repair method and the $\ell^1$-BA repair method. These two methods are separately applied to the same set of OTC FX data as in Figure \ref{fig:detectionTsOtc}, and the histograms are plotted by stacking data of all ten currency pairs and historical dates.}
\label{fig:statsL1L1BA}
\end{figure}

A detailed example showing how the two repair methods work in reality is given in Figure \ref{fig:exampleL1BA}. From left to right, the displayed data are ordered by increasing strikes, grouped by expiry. The light blue areas are confined by bid-ask spread as a percentage of option prices (green lines). We see that ITM and OTM options have wider bid-ask spreads than ATM options do. The $\ell^1$-BA repair method results in fewer effective perturbations. First, there is one less effective perturbation of 1M option prices, at the cost of perturbing a few 2W, 3W and 1M option prices to their bid or ask prices. Second, all four effective perturbations of 4M option prices by the $\ell^1$-norm repair are replaced by six ineffective perturbations of 6M option prices by the $\ell^1$-BA repair.

\begin{figure}[!h]
\centering
\includegraphics[scale=0.59]{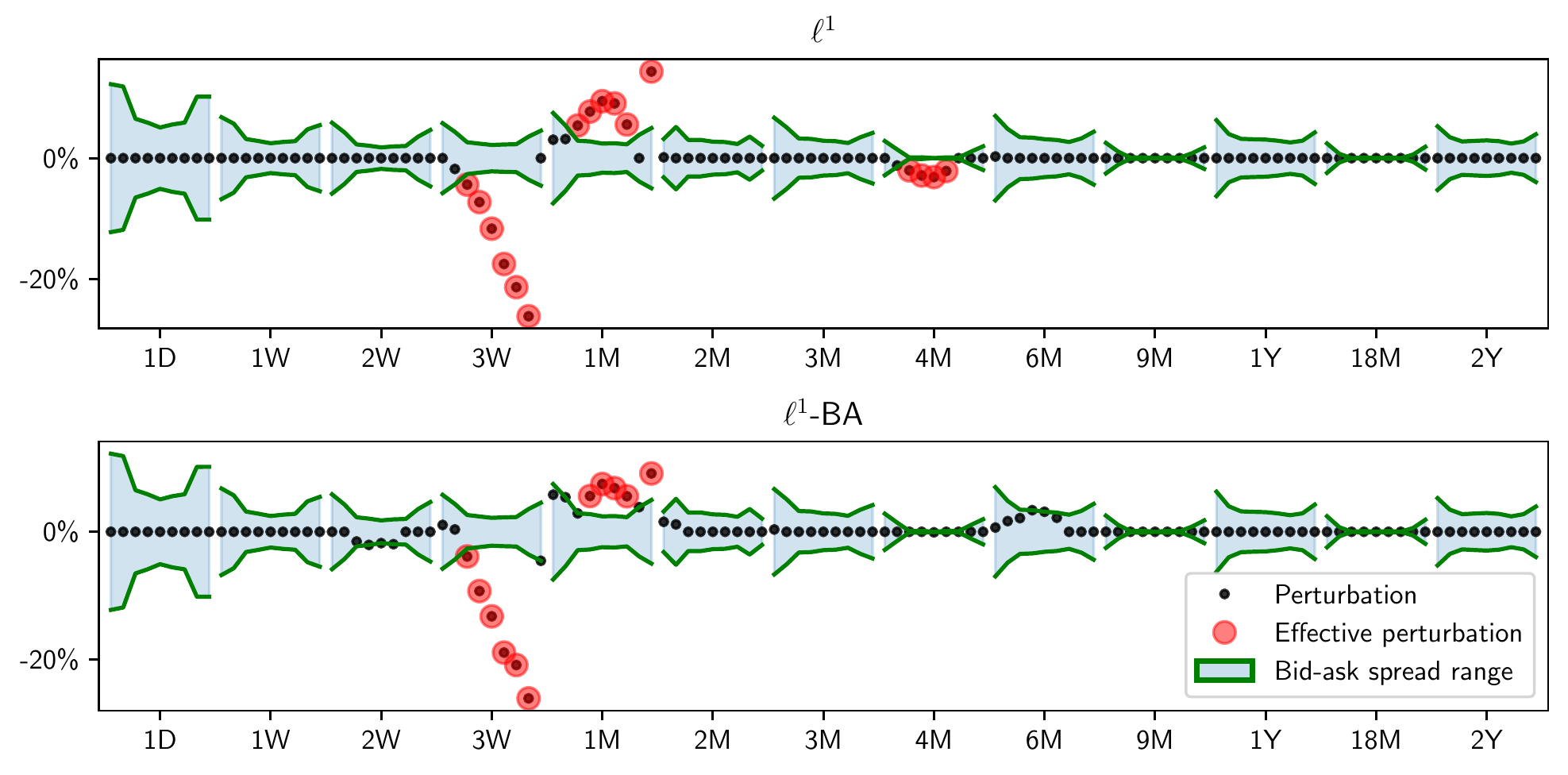}
\caption{Perturbations (as percentages of the raw price data) resulted from the $\ell^1$-norm and the $\ell^1$-BA objectives. Data used are bid, ask and mid prices for OTC-traded USDBRL options on 18th September, 2008.}
\label{fig:exampleL1BA}
\end{figure}

For a given set of prices, if none of the perturbations is effective, then the bid and ask quotes given by the market admit some arbitrage-free prices that fall within the bid-ask price bounds. In contrast, effective perturbations imply the existence of \textit{executable} arbitrages that are exploitable through matching existing bid or ask orders in the market, see Section \ref{sec:include_bid_ask}. In Table \ref{tab:true_arbitrage}, we count the number of days when there is arbitrage in mid-prices ($N^\varepsilon > 0$) and the number of days when there is executable arbitrage ($N^{\varepsilon, \delta} > 0$) in historical data for the four currency pairs that have been seen to have the most occurrences of arbitrages.


\begin{table}[!h]
\centering
\footnotesize
\begin{tabular}{ccccc}
\toprule
Currency pair                 & EURGBP & USDBRL & USDKRW & USDMXN \\
\cmidrule(lr){1-1} \cmidrule(lr){2-5}
\#days when $N^\varepsilon > 0$           & 470    & 708    & 623    & 577    \\
\#days when $N^{\varepsilon, \delta} > 0$ & 285    & 89     & 163    & 144 \\
\bottomrule
\end{tabular}
\caption{Number of days when there is arbitrage in mid-prices ($N^\varepsilon > 0$) and when there is executable arbitrage ($N^{\varepsilon, \delta} > 0$).}
\label{tab:true_arbitrage}
\end{table}

\vspace{-4mm}
\subsection{Improvement to model calibration}

We verify that our repair method improves model calibration with more robust parameter estimates and smaller calibration error.

\subsubsection*{Test framework}

Let $\Theta$ be model parameters. We specify $\Theta = \overline{\Theta}$ and generate model prices $\mathbf{c}$ for call options on a set of expiries and strikes. Then we carry out the following steps $M$ times. For the $m$-th time:
\begin{enumerate}[label=(\arabic*)]
    \item Simulate noises to create synthetic arbitrageable price data $\mathbf{\tilde{c}}^{(m)}$, following the method in Section \ref{sec:stress_testing}. Recall that $\lambda \in (0,1]$ portion of prices are polluted by Gaussian noises of variance $\sigma^2_\xi$.
    \item Repair arbitrage in $\mathbf{\tilde{c}}^{(m)}$ to get arbitrage-free data $\mathbf{\hat{c}}^{(m)}$.
    \item Calibrate model parameters $\Theta$ to $\mathbf{\tilde{c}}^{(m)}$ and $\mathbf{\hat{c}}^{(m)}$ separately\footnote{Note that we must apply exactly the same numerical procedure for these two separate calibrations, i.e. the same optimisation algorithm, terminal criteria, lower and upper bounds, and initial values.}, and get calibrated parameters $\widetilde{\Theta}^{(m)}$ and $\widehat{\Theta}^{(m)}$, respectively. Defining the calibration objective as $G(\Theta; \mathbf{c}) = \sum_{j=1}^N ( c_j^\Theta - c_j )^2$ where $c_j^\Theta$ is the model price for the $j$-th option, we have
    \begin{equation*}
        \widetilde{\Theta}^{(m)} = \argmin_{\Theta} G(\Theta; \mathbf{\tilde{c}}^{(m)}), \quad \widehat{\Theta}^{(m)} = \argmin_{\Theta} G(\Theta; \mathbf{\hat{c}}^{(m)}).
    \end{equation*}
\end{enumerate}

We measure model calibration performance by two metrics, which are (a) the \textit{robustness} defined by variations in the parameter estimates, and (b) the calibration error defined as the square root of the minimal objective value. Since we have parameter estimates $\{\widetilde{\Theta}^{(m)}\}_{1 \leq m \leq M}$ and $\{\widehat{\Theta}^{(m)}\}_{1 \leq m \leq M}$, we can compare the variations in them for assessing robustness. For each $m$, we define the (relative) reduction of calibration error as
\begin{equation*}
    \Delta G^{(m)} = 1 - \sqrt{\frac{G(\widehat{\Theta}^{(m)}; \mathbf{\hat{c}}^{(m)})}{G(\widetilde{\Theta}^{(m)}; \mathbf{\tilde{c}}^{(m)})}}.
\end{equation*}

\subsubsection*{Heston model calibration}

We carry out a test on calibration of the Heston model \cite{Heston1993}. Recall that the Heston model is described by the SDEs with model parameters $\Theta = (\nu_0, \theta, k, \sigma, \rho)$:
\begin{align*}
    \diff S_t & = r_t S_t \diff t + \sqrt{\nu_t} S_t \diff W_t^S, \\
    \diff \nu_t & = k (\theta - \nu_t) \diff t + \sigma \sqrt{\nu_t} \diff W_t^\nu, ~
    \diff \langle W_t^S, W_t^\nu \rangle = \rho \diff t,
\end{align*}
where the Feller condition $2k\theta > \sigma^2$ is sufficient to ensure strict positivity of the instantaneous variance process $\nu_t$.

We specify a typical set of expiries and strikes that is observed on a day in the OTC market, such as the one shown in Figure \ref{fig:sampleDataOTC}. Other simulation parameters and ground truth model parameters\footnote{Heston model parameters are chosen as those that reproduce a typical call price surface for USDBRL options. Noise simulation parameters $\lambda$ and $\sigma_\xi$ are chosen to mimic severe but not extreme arbitrage scenarios (measured by the fraction of perturbed prices by the repair method) observed in real world data.} are listed in Table \ref{tab:params_values}.
\begin{table}[!h]
\centering
\footnotesize
\begin{tabular}{cccccccccc}
\toprule
\multicolumn{1}{l}{} & \multicolumn{5}{c}{Heston model}              & \multicolumn{4}{c}{Simulation}       \\ \cmidrule(lr){1-1} \cmidrule(lr){2-6} \cmidrule(lr){7-10}
Parameter            & $\nu_0$ & $\theta$ & $k$  & $\sigma$ & $\rho$ & $N$ & $M$ & $\lambda$ & $\sigma_\xi$ \\
Value                & 0.003   & 0.008    & 2.32 & 0.38     & 0.36   & 117 & 500 & 0.25      & 0.1 or 1 \\       
\bottomrule
\end{tabular}
\caption{Parameter values}
\label{tab:params_values}
\end{table}

\begin{figure}[!h]
\centering
\begin{subfigure}[b]{\textwidth}
    \includegraphics[scale=0.59]{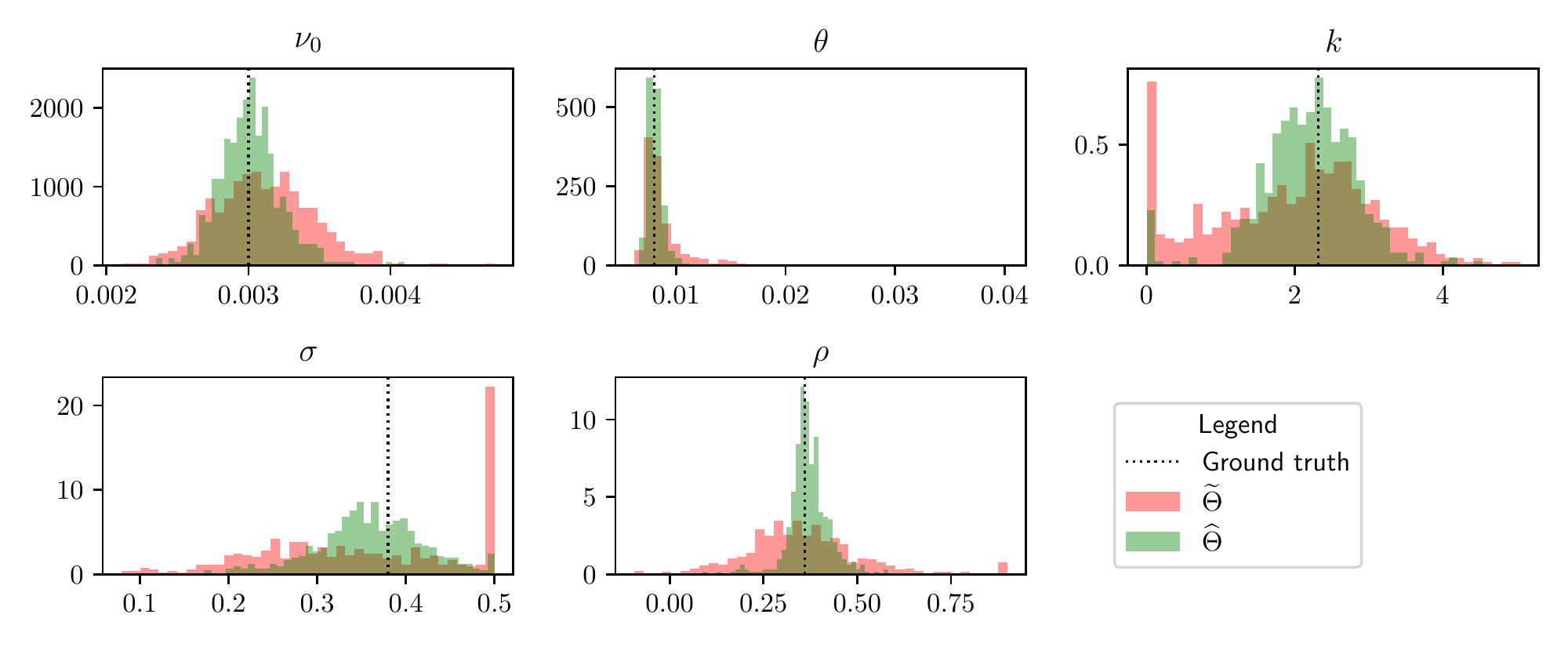}
    \vspace{-2mm}
    \caption{$\sigma_\xi=0.1$.}
\end{subfigure} \\
\begin{subfigure}[b]{\textwidth}
    \includegraphics[scale=0.59]{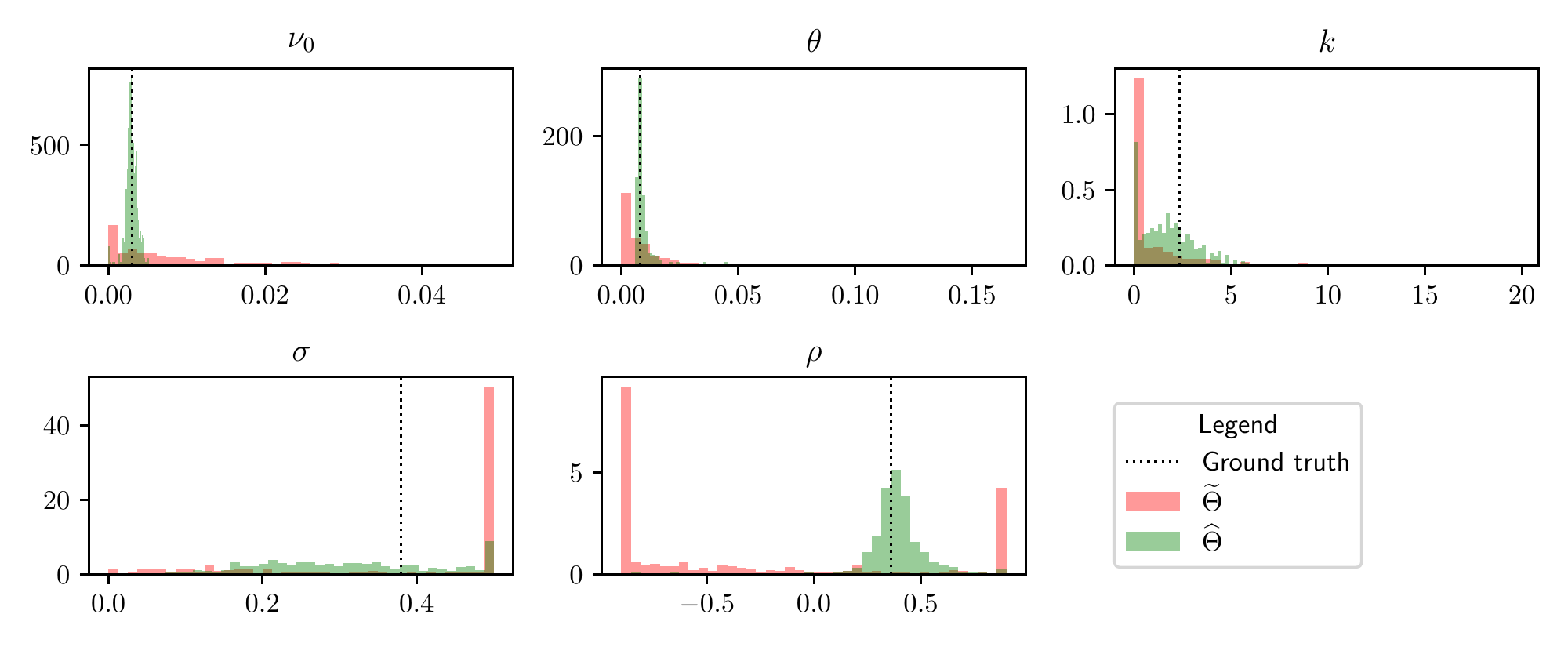}
    \vspace{-2mm}
    \caption{$\sigma_\xi=1$.}
\end{subfigure} \\
\caption{Sample (normed) histograms of $\widetilde{\Theta}$ and $\widehat{\Theta}$, where $\Theta = (\nu_0,  \theta, k, \sigma, \rho)$.}
\vspace{-4mm}
\label{fig:synHestonParams}
\end{figure}

Next, we follow the test framework and evaluate $\widetilde{\Theta}^{(m)}$, $\widehat{\Theta}^{(m)}$, $\Delta G^{(m)}$ and $\Delta t^{(m)}$ for $m = 1,\dots,M$. In Figure \ref{fig:synHestonParams}, we plot and compare the normed histograms of calibrated Heston parameters $\widetilde{\Theta}$ (using noisy data) and  $\widehat{\Theta}$ (using repaired data) given different choices of $\sigma_\xi$. The ground truth parameter values are also indicated by vertical dotted lines. Repairing data does make the model calibration more robust, as supported by two types of evidence. First, there are apparently more variations in $\widetilde{\Theta}$ than in $\widehat{\Theta}$. Second, $\widetilde{\Theta}$ tends to hit the bounds set in the numerical optimisation procedure (e.g. 0 for $k$, 0.5 for $\sigma$, 1 for $\rho$) much more often than $\widehat{\Theta}$ does. Moreover, when the price data are more noisy (larger $\sigma_\xi$) so that more prices with arbitrage are present, the robustness improvement of model calibration by the repair method becomes more significant.

In Figure \ref{fig:synHestonReduction}, we plot the histograms of $\Delta G$ and indicate their means by vertical dotted lines. Repairing arbitrage in data reduces the calibration errors in all $M$ simulations with no exception. Moreover, the more noisy the raw data are, the arbitrage repair method reduces relatively more calibration errors. On average, repairing data can reduce the calibration error by more than $70\%$ for $\sigma_\xi = 0.1$, and more than $95\%$ for $\sigma_\xi = 1$.
\begin{figure}[!h]
    \centering
    \vspace{-4mm}
    \includegraphics[scale=.59]{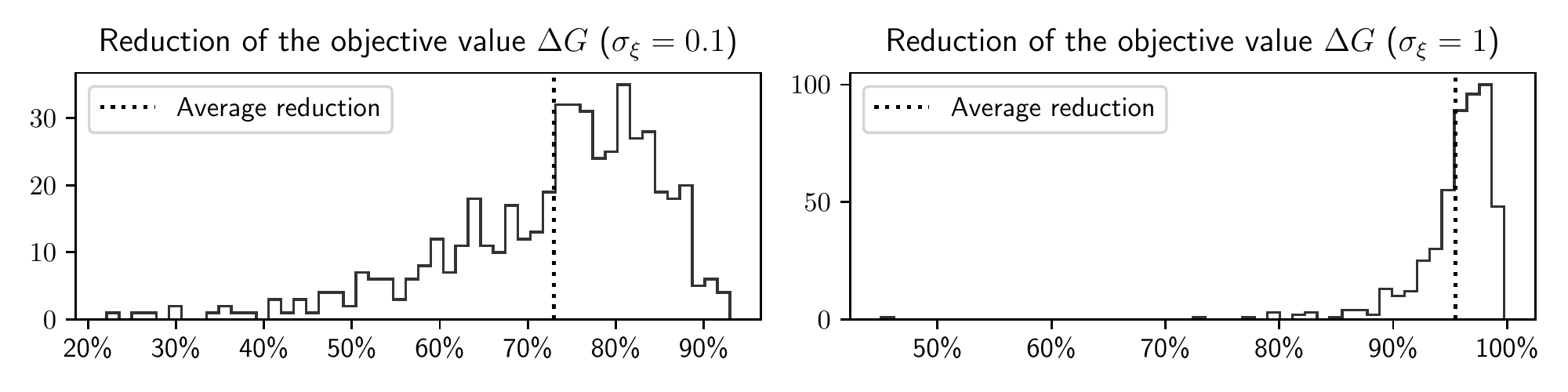}
    \caption{Sample histograms of (relative) reductions in calibration error.}
    \vspace{-4mm}
    \label{fig:synHestonReduction}
\end{figure}

Hence, for model calibration task, there is more benefit of repairing data by removing arbitrage when the data contain larger noise.

\subsection{Identifying intra-day executable arbitrage}

We can use the $\ell^1$-BA repair method on order book data for identifying executable arbitrage. An example is given in Figure \ref{fig:hf}. We collect the order book data for all E-mini S\&P 500 monthly European call options from 12:00 ET to 16:10 ET on 12th June, 2020.


\begin{figure}[!h]
    \centering
    \includegraphics[scale=.59]{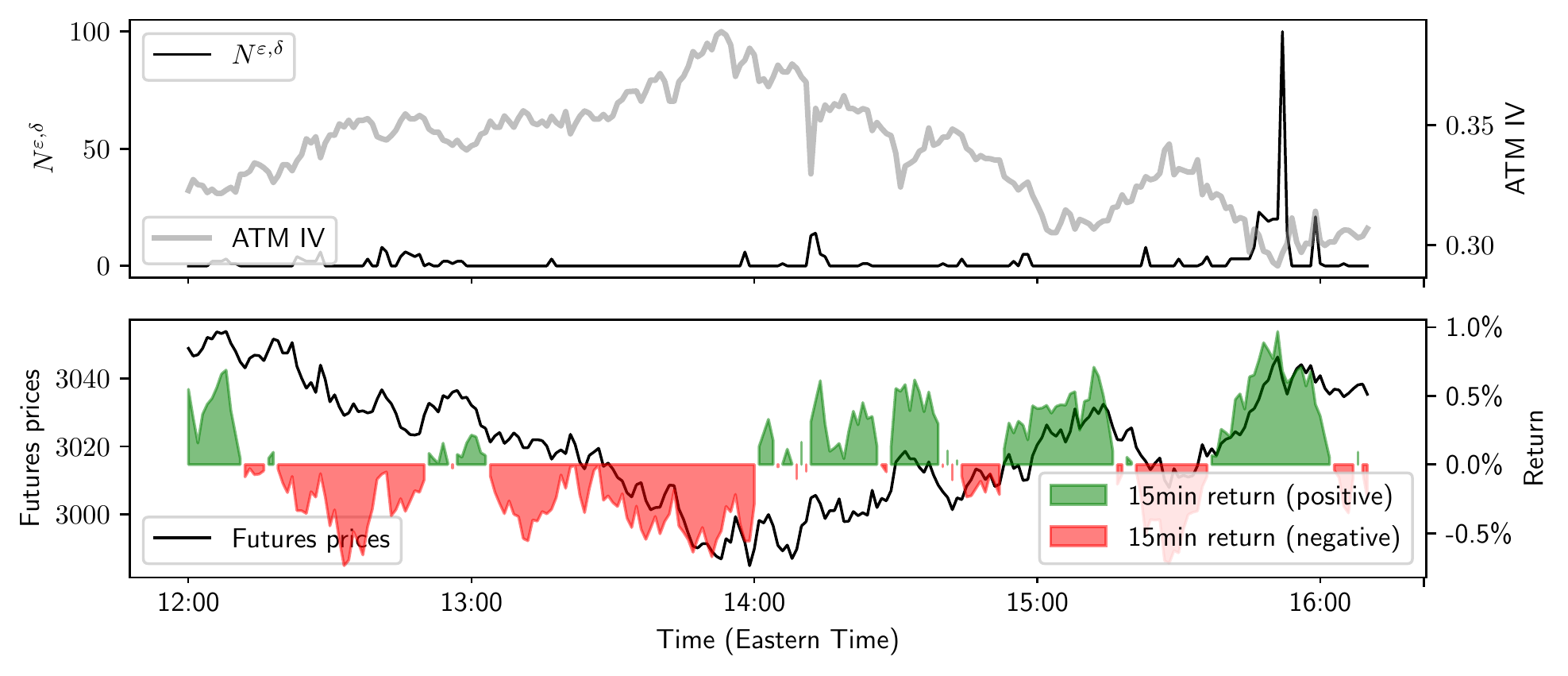}
    \caption{\textit{Top} -- The formation and disappearance of intra-day executable arbitrage opportunities in the E-mini S\&P 500 monthly European call option market on 12th June, 2020. \textit{Bottom} -- front-month futures' prices and 15-minute return.}
    \label{fig:hf}
\end{figure}

We extract the active best ask and best bid prices for all quoted call options from the order book at the end of every minute. Then we compute mid prices, apply the $\ell^1$-BA repair method to the mid prices, and count the number of effective perturbations $N^{\varepsilon, \delta}$. Recall from Section \ref{sec:exe_arbitrage} that, given $\delta_0$ small, there exists executable arbitrage if the $\ell^1$-BA repair method results in effective perturbations. In the top plot of Figure \ref{fig:hf}, the black line gives $N^{\varepsilon, \delta}$ over time, while we also indicate the ATM implied volatility of the front-month option (which would expire on 19th June) by the grey line. The bottom plot gives the prices and the 15-minute returns of the front-month futures contract. The downward trend of the futures market was inverted around 14:00 ET, after when the implied volatility also falls gradually from its peak. 

There is a large spike of $N^{\varepsilon, \delta}$ at around 15:52 ET, a few minutes before the close of the S\&P 500 index market at 16:00 ET. This spike coincided with rallies in the futures market, while the IV maintained its relatively low level. There are some clusters of smaller spikes of $N^{\varepsilon, \delta}$ outside of the US trading hours. Apart from this arbitrage outbreak preceding the close of the underlying market, which lasted for around 15 minutes, there seems to be trivial executable arbitrage during the rest of the afternoon trading hours, even when the market underwent regime switch (from downward trend to upward trend) at around 14:00 ET. 

\appendix
\section{Localisation of static arbitrage constraints}
\label{appendix:localisationConstraints}

To localise calendar butterfly constraints, we use a sequential build-up of local constraints from the shortest expiry to the longest expiry. Define $\mathcal{D}_{i} := \{ (k_j^i, c_j^i) : 1 \leq j \leq n_i\}$ as price data for options of expiry $T_i \in \mathcal{T}^e$. Given arbitrage-free $\mathcal{D}_{i^*}$, we construct constraints such that adding price data of any longer-expiry option should not introduce arbitrage. This is done locally in two steps, where we scan a neighbourhood of each $k_j^{i^*}$. The first step, we call ``absolute location convexity'' \hyperlink{c61}{C6.1}, finds constraints ensuring that adding \textit{any single} data point $(k^i_j, c^i_j)$ where $i>i^*$ will not introduce arbitrage. In Figure \ref{fig:absConvexity} we indicate the regions where adding a single data point will not introduce arbitrage for four types of strike neighbourhood. In the second step ``relative location convexity'' \hyperlink{c62}{C6.2}, we find constraints making sure that adding \textit{all} data points $(k^i_j, c^i_j)$ where $i>i^*$ will not introduce arbitrage for two types of strike neighbourhood. As shown in Figure \ref{fig:relConvexity}, if we draw line segments by linking each added point and the reference point $o = (k^{i^*}_{j^*}, c^{i^*}_{j^*})$, we require the slope of any line on the left $\{l_i\}$ to be not greater than the slope of any line on the right $\{r_j\}$.

\begin{figure}[!h]
\centering
  	\begin{subfigure}[b]{0.24\linewidth}
    	\begin{tikzpicture}[scale=1][
   	thick,
    >=stealth']
  	
  	\coordinate (O) at (0,0);

 	\draw[->] (-0.1,0) -- (2.6,0) coordinate[label = {below:$k$}] (xmax);
  	\draw[->] (0,-0.1) -- (0,3) coordinate[label = {right:$c$}] (ymax);	
	
	\node (a) [circle, fill=black, minimum size=3pt, inner sep=0pt] at (0,2.3) {};
	\node (b) [circle, fill=black, minimum size=3pt, inner sep=0pt] at (.5,1.1) {};
	\node (c) [circle, fill=black, minimum size=3pt, inner sep=0pt] at (1.3,.6) {};
	\node (d) [circle, fill=black, minimum size=3pt, inner sep=0pt] at (2.3,.4) {};
	
	\draw[gray] (a)--(b);
	\draw[gray] (b)--(c);
	\draw[gray] (c)--(d);
	
	
	\coordinate (b1) at (b);
	\coordinate (c1) at (c);
	\coordinate (i) at (intersection of O--(0,1) and b1--c1);
	\path let \p1=(b1) in coordinate (b2) at (\x1, 2.8);
 	\fill[green!50, fill opacity=.5] (i) -- (0,2.8) -- (b2) -- (b1) -- cycle;
 	
 	\draw[dashed] (b1) -- (b1 |- O) node[label = {below:$k^{i^*}_{j^*}$}] {};
\end{tikzpicture}
    \caption{$j^*=1$} \label{fig:cvxWeak1}  
  	\end{subfigure}
  	\begin{subfigure}[b]{0.24\linewidth}
	\begin{tikzpicture}[scale=1][
   	thick,
    >=stealth']
  	
  	\coordinate (O) at (0,0);

 	\draw[->] (-0.1,0) -- (2.6,0) coordinate[label = {below:$k$}] (xmax);
  	\draw[->] (0,-0.1) -- (0,3) coordinate[label = {right:$c$}] (ymax);	
	
	\node (a) [circle, fill=black, minimum size=3pt, inner sep=0pt] at (0,2.3) {};
	\node (b) [circle, fill=black, minimum size=3pt, inner sep=0pt] at (.5,1.1) {};
	\node (c) [circle, fill=black, minimum size=3pt, inner sep=0pt] at (1.3,.6) {};
	\node (d) [circle, fill=black, minimum size=3pt, inner sep=0pt] at (2.3,.4) {};
	
	\draw[gray] (a)--(b);
	\draw[gray] (b)--(c);
	\draw[gray] (c)--(d);
	
	\coordinate (a1) at (a);
	\coordinate (b1) at (b);
	\coordinate (c1) at (c);
	\coordinate (d1) at (d);
	\coordinate (i) at (intersection of a1--b1 and c1--d1);
	\path let \p1=(b1) in coordinate (b2) at (\x1, 2.8);
	\path let \p1=(c1) in coordinate (c2) at (\x1, 2.8);
 	\fill[green!50, fill opacity=.5] (b1) -- (b2) -- (c2) -- (c1) -- (i) -- cycle;
 	
 	\draw[dashed] (b1) -- (b1 |- O) node[label = {below:$k^{i^*}_{j^*-1}$}] {};
 	\draw[dashed] (c1) -- (c1 |- O) node[label = {below:$k^{i^*}_{j^*}$}] {};
 	\draw[dashed] (d1) -- (d1 |- O) node[label = {below:$k^{i^*}_{j^*+1}$}] {};
 	
\end{tikzpicture}
\caption{$j^* \in [2, n_{i^*}-1]$} \label{fig:cvxWeak2}
  	\end{subfigure}
  	\begin{subfigure}[b]{0.24\linewidth}
    	\begin{tikzpicture}[scale=1][
   	thick,
    >=stealth']
  	
  	\coordinate (O) at (0,0);

 	\draw[->] (-0.1,0) -- (2.6,0) coordinate[label = {below:$k$}] (xmax);
  	\draw[->] (0,-0.1) -- (0,3) coordinate[label = {right:$c$}] (ymax);	
	
	\node (a) [circle, fill=black, minimum size=3pt, inner sep=0pt] at (0,2.3) {};
	\node (b) [circle, fill=black, minimum size=3pt, inner sep=0pt] at (.5,1.1) {};
	\node (c) [circle, fill=black, minimum size=3pt, inner sep=0pt] at (1.3,.6) {};
	\node (d) [circle, fill=black, minimum size=3pt, inner sep=0pt] at (2.3,.4) {};
	
	\draw[gray] (a)--(b);
	\draw[gray] (b)--(c);
	\draw[gray] (c)--(d);
	
	\coordinate (a1) at (a);
	\coordinate (b1) at (b);
	\coordinate (c1) at (c);
	\coordinate (d1) at (d);
	\path let \p1=(d1) in coordinate (d2) at (0, \y1);
	\coordinate (i) at (intersection of b1--c1 and d2--d1);
	\path let \p1=(c1) in coordinate (c2) at (\x1, 2.8);
	\path let \p1=(d1) in coordinate (d3) at (\x1, 2.8);
 	\fill[green!50, fill opacity=.5] (c1) -- (c2) -- (d3) -- (d1) -- (i) -- cycle;
 	
 	\draw[dashed] (c1) -- (c1 |- O) node[label = {below:$k^{i^*}_{j^*-1}$}] {};
 	\draw[dashed] (d1) -- (d1 |- O) node[label = {below:$k^{i^*}_{j^*}$}] {};
 	
\end{tikzpicture}
    \caption{$j^*=n_{i^*}$} \label{fig:cvxWeak3}
  	\end{subfigure}
  	\begin{subfigure}[b]{0.24\linewidth}
    	\begin{tikzpicture}[scale=1][
   	thick,
    >=stealth']
  	
  	\coordinate (O) at (0,0);

 	\draw[->] (-0.1,0) -- (3,0) coordinate[label = {below:$k$}] (xmax);
  	\draw[->] (0,-0.1) -- (0,3) coordinate[label = {right:$c$}] (ymax);	
	
	\node (a) [circle, fill=black, minimum size=3pt, inner sep=0pt] at (0,2.3) {};
	\node (b) [circle, fill=black, minimum size=3pt, inner sep=0pt] at (.5,1.1) {};
	\node (c) [circle, fill=black, minimum size=3pt, inner sep=0pt] at (1.3,.6) {};
	\node (d) [circle, fill=black, minimum size=3pt, inner sep=0pt] at (2.3,.4) {};
	
	\draw[gray] (a)--(b);
	\draw[gray] (b)--(c);
	\draw[gray] (c)--(d);
	
	\coordinate (a1) at (a);
	\coordinate (b1) at (b);
	\coordinate (c1) at (c);
	\coordinate (d1) at (d);
	\coordinate (e1) at (2.8,0);
	\coordinate (e2) at (2.8,1);
	\coordinate (i) at (intersection of c1--d1 and e1--e2);
	\path let \p1=(i) in coordinate (i1) at (\x1, 2.8);
	\path let \p1=(d1) in coordinate (d2) at (\x1, 2.8);
 	\fill[green!50, fill opacity=.5] (d1) -- (d2) -- (i1) -- (i) -- cycle;
 	
 	\draw[dashed] (c1) -- (c1 |- O) node[label = {below:$k^{i^*}_{j^*-1}$}] {};
 	\draw[dashed] (d1) -- (d1 |- O) node[label = {below:$k^{i^*}_{j^*}$}] {};
 	
\end{tikzpicture} 
    \caption{$j^*=n_{i^*}$, $k>k^{i^*}_{j^*}$} \label{fig:cvxWeak4}
  	\end{subfigure}
\caption{Absolute location convexity constraint, discussed in four cases. Points falling in the green region satisfy the absolute location convexity constraint.}
\label{fig:absConvexity}
\end{figure}
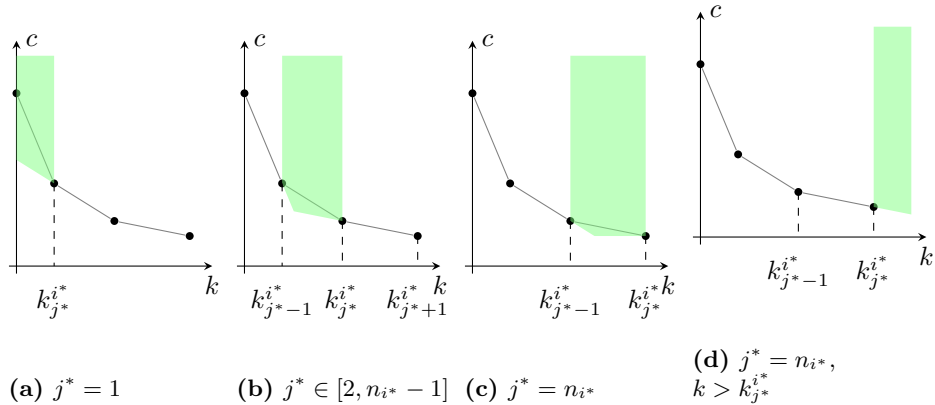

\begin{figure}[!h]  
\centering 
  \begin{subfigure}[b]{0.49\linewidth}
    \begin{tikzpicture}[scale=1][
   	thick,
    >=stealth']
  	
  	\coordinate (O) at (0,0);

 	\draw[->] (-0.1,0) -- (5.6,0) coordinate[label = {below:$k$}] (xmax);
  	\draw[->] (0,-0.1) -- (0,4) coordinate[label = {right:$c$}] (ymax);	
	
	\node (a) [circle, fill=black, minimum size=3pt, inner sep=0pt] at (0,3.5) {};
	\node (b) [circle, fill=black, minimum size=3pt, inner sep=0pt] at (.3,2.5) {};
	\node (c) [circle, fill=black, minimum size=3pt, inner sep=0pt, label={left:$o$}] at (1.8,1.2) {};
	\node (d) [circle, fill=black, minimum size=3pt, inner sep=0pt] at (4,.5) {};
	
	\draw[gray] (a)--(b);
	\draw[gray] (b)--(c);
	\draw[gray] (c)--(d);
	
	\coordinate (a1) at (a);
	\coordinate (b1) at (b);
	\coordinate (c1) at (c);
	\coordinate (d1) at (d);
	\path let \p1=(d1) in coordinate (d2) at (0, \y1);
	\coordinate (i1) at (intersection of a1--b1 and c1--d1);
	\coordinate (i2) at (intersection of b1--c1 and d1--d2);
	\path let \p1=(b1) in coordinate (b2) at (\x1, 3.8);
	\path let \p1=(d1) in coordinate (d3) at (\x1, 3.8);
	\fill[green!50, fill opacity=.5] (b1) -- (b2) -- (d3) -- (d1) -- (i2) -- (c1) -- (i1) -- cycle;
 	
	\node (l1) [circle, fill=red, minimum size=3pt, inner sep=0pt, label={left:$l_2$}] at (1,1.6) {};
	\node (l2) [circle, fill=red, minimum size=3pt, inner sep=0pt, label={left:$l_1$}] at (.9,3) {};
	\node (l3) [circle, fill=red, minimum size=3pt, inner sep=0pt] at (.7,2) {};
	
	\node (r1) [circle, fill=blue, minimum size=3pt, inner sep=0pt, label={right:$r_1$}] at (2.5,.9) {};
	\node (r2) [circle, fill=blue, minimum size=3pt, inner sep=0pt] at (2.3,3.5) {};
	\node (r3) [circle, fill=blue, minimum size=3pt, inner sep=0pt, label={right:$r_2$}] at (2.8,.55) {};
	\node (r4) [circle, fill=blue, minimum size=3pt, inner sep=0pt] at (3.8,1.5) {};
	
	\draw[red!60, densely dotted] (l1)--(c);
	\draw[red!60, densely dotted] (l2)--(c);
	\draw[red!60, densely dotted] (l3)--(c);
	
	\draw[blue!60, densely dotted] (r1)--(c);
	\draw[blue!60, densely dotted] (r2)--(c);
	\draw[blue!60, densely dotted] (r3)--(c);
	\draw[blue!60, densely dotted] (r4)--(c);
 	
 	\draw[dashed] (c1) -- (c1 |- O) node[label = {below:$k^{i^*}_{j^*}$}] {};
\end{tikzpicture}
    \caption{$j^* \in [1, n_{i^*}-1]$} \label{fig:cvxStrong1}  
  \end{subfigure}
  \begin{subfigure}[b]{0.49\linewidth}
	\begin{tikzpicture}[scale=1][
   	thick,
    >=stealth']
  	
  	\coordinate (O) at (0,0);

 	\draw[->] (-0.1,0) -- (6,0) coordinate[label = {below:$k$}] (xmax);
  	\draw[->] (0,-0.1) -- (0,4) coordinate[label = {right:$c$}] (ymax);	
	
	\node (a) [circle, fill=black, minimum size=3pt, inner sep=0pt] at (0,3.5) {};
	\node (b) [circle, fill=black, minimum size=3pt, inner sep=0pt] at (.3,2.5) {};
	\node (c) [circle, fill=black, minimum size=3pt, inner sep=0pt] at (1.8,1.2) {};
	\node (d) [circle, fill=black, minimum size=3pt, inner sep=0pt, label={below:$o$}] at (4,.5) {};
	
	\draw[gray] (a)--(b);
	\draw[gray] (b)--(c);
	\draw[gray] (c)--(d);
	
	\coordinate (a1) at (a);
	\coordinate (b1) at (b);
	\coordinate (c1) at (c);
	\coordinate (d1) at (d);
	\path let \p1=(d1) in coordinate (d2) at (0, \y1);
	\coordinate (i1) at (intersection of b1--c1 and d1--d2);
	\path let \p1=(c1) in coordinate (c2) at (\x1, 3.8);
	\path let \p1=(d1) in coordinate (d3) at (\x1, 3.8);
	\fill[green!50, fill opacity=.5] (c1) -- (c2) -- (d3) -- (d1) -- (i1) -- cycle;
 	
	\coordinate (e1) at (5.5,0);
	\coordinate (e2) at (5.5,1);
	\coordinate (i2) at (intersection of c1--d1 and e1--e2);
	\path let \p1=(i2) in coordinate (i3) at (\x1, 3.8);
	\fill[green!50, fill opacity=.5] (d1) -- (d3) -- (i3) -- (i2) -- cycle;	
 	
	\node (l1) [circle, fill=red, minimum size=3pt, inner sep=0pt, label={left:$l_1$}] at (2.5,.9) {};
	\node (l2) [circle, fill=red, minimum size=3pt, inner sep=0pt] at (2.3,3.5) {};
	\node (l3) [circle, fill=red, minimum size=3pt, inner sep=0pt, label={left:$l_2$}] at (3,.6) {};
	\node (l4) [circle, fill=red, minimum size=3pt, inner sep=0pt] at (3.8,1.5) {};
	
	\node (r1) [circle, fill=blue, minimum size=3pt, inner sep=0pt, label={right:$r_1$}] at (4.5,.9) {};
	\node (r2) [circle, fill=blue, minimum size=3pt, inner sep=0pt, label={right:$r_2$}] at (4.7,.35) {};
	\node (r3) [circle, fill=blue, minimum size=3pt, inner sep=0pt] at (5.3,.5) {};
	
	\draw[red!60, densely dotted] (l1)--(d);
	\draw[red!60, densely dotted] (l2)--(d);
	\draw[red!60, densely dotted] (l3)--(d);
	\draw[red!60, densely dotted] (l4)--(d);
	
	\draw[blue!60, densely dotted] (r1)--(d);
	\draw[blue!60, densely dotted] (r2)--(d);
	\draw[blue!60, densely dotted] (r3)--(d);
 	
 	\draw[dashed] (d1) -- (d1 |- O) node[label = {below:$k^{i^*}_{j^*}$}] {};
\end{tikzpicture}
	\caption{$j^* = n_{i^*}$} \label{fig:cvxStrong2}
  \end{subfigure}
\caption{Relative location convexity constraint, discussed in two cases. Points within the green region satisfy the absolute location convexity constraint.}
\label{fig:relConvexity}
\end{figure}
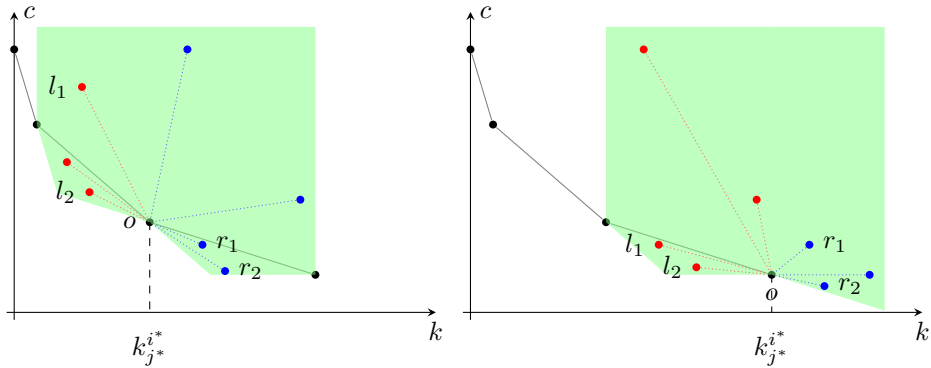

\section{Proof of Proposition \ref{prop:reducedConstraint}}
\label{appendix:proofReducedConstraint}

We prove Proposition \ref{prop:reducedConstraint} by establishing Lemma \ref{lemma:arbitrage1}, \ref{lemma:arbitrage2} and \ref{lemma:arbitrage3}.

\begin{lemma}\label{lemma:arbitrage1}
If \hyperlink{c1}{C1}, \hyperlink{c2}{C2} and \hyperlink{c3}{C3} are satisfied, then all outrights, vertical spreads and vertical butterflies are non-negative. In addition, all test vertical spreads are not greater than 1.
\end{lemma}

\begin{proof}
We consider the prices of call options with the same expiry $T_i$ where $i \in [1,m]$.

First, we prove that any vertical spread is non-negative, i.e. $\forall 0 \leq j_1 < j_2 \leq n_i$, $c^i_{j_1} \geq c^i_{j_2}$. This is true by the vertical spread constraint \hyperlink{c2}{C2}, as $c^i_{j_1} \geq c^i_{j_1+1} \geq \cdots \geq c^i_{j_2}$.

Second, we show that all outrights are non-negative, i.e. $\forall j \in [0, n_i]$, $c^i_j\geq 0$. Given the outright constraint \hyperlink{c1}{C1} and that any vertical spread is non-negative, we have $c^i_j \geq c^i_{n_i} \geq 0$.

Next, we show that any vertical butterfly is non-negative, i.e. $\forall 0 \leq j_1 < j < j_2 \leq n_i$, $\beta(i,j_2;i,j) - \beta(i,j;i,j_1) \geq 0$. To do that, we claim
\begin{subequations}
\begin{equation}
\beta(i,j_1+1;i,j_1) \leq \beta(i,j;i,j_1+1), \quad \text{if } j_1 < j-1,
\label{eq:claimj1}
\end{equation}
\begin{equation}
\beta(i,j_2;i,j_2-1) \geq \beta(i,j_2-1;i,j), \quad \text{if } j_2 > j+1.
\label{eq:claimj2}
\end{equation}
\end{subequations}
These two claims can be proved by induction. Here we only show the proof for (\ref{eq:claimj1}). It is true that $\beta(i,j-1;i,j-2) \leq \beta(i,j;i,j-1)$ (the $j_1=j-2$ case for (\ref{eq:claimj1})) by the vertical butterfly constraint \hyperlink{c3}{C3}. Assume that (\ref{eq:claimj1}) holds for $j_1=l < j-2$, i.e.
\begin{equation*}
\begin{split}
c^i_j \geq c^i_{l+1} + (k^i_j-k^i_{l+1}) \frac{c^i_{l+1}-c^i_{l}}{k^i_{l+1}-k^i_{l}} & = c^i_{l+1} + \left[ (k^i_j-k^i_l)-(k^i_{l+1}-k^i_l) \right] \frac{c^i_{l+1}-c^i_{l}}{k^i_{l+1}-k^i_{l}} \\
& = c^i_l + (k^i_j - k^i_l) \beta(i,l+1;i,l).
\end{split}
\end{equation*}
This leads to $\beta(i,j;i,l) \geq \beta(i,l+1;i,l)$. Again by \hyperlink{c3}{C3} we have $\beta(i,l+1;i,l) \geq \beta(i,l;i,l-1)$. Hence, $\beta(i,j;i,l) \geq \beta(i,l;i,l-1)$, which is the stated inequality by setting $j_1=l-1$. Therefore, (\ref{eq:claimj1}) holds by induction in reverse order from $j_1=j-2$ to $j_1=0$. Thereafter, (\ref{eq:claimj1}) implies
\begin{equation*}
\begin{split}
-c^i_{j_1} & \leq -c^i_{j_1+1} + (k^i_{j_1+1} - k^i_{j_1}) \frac{c^i_j - c^i_{j_1+1}}{k^i_j - k^i_{j_1+1}} \\ 
& = -c^i_{j_1+1} + \left[ (k^i_{j_1+1}-k^i_j)-(k^i_{j_1}-k^i_j) \right] \frac{c^i_j - c^i_{j_1+1}}{k^i_j - k^i_{j_1+1}} \\
& = -c^i_{j} + (k^i_{j_1+1}-k^i_j) \beta(i,j;i,j_1+1),
\end{split}
\end{equation*}
which leads to $\beta(i,j;i,j_1) \leq \beta(i,j;i,j_1+1)$. Similarly we have $\beta(i,j;i,j_1) \leq \beta(i,j;i,j_1+1) \leq \cdots \leq \beta(i,j;i,j-1)$. In similar fashion we can prove (\ref{eq:claimj2}) by induction from $j_1=j+2$ to $j_1=n_i$, and deduce $\beta(i,j_2;i,j) \geq \beta(i,j_2-1;i,j) \geq \cdots \geq \beta(i,j+1;i,j)$. Therefore, with \hyperlink{c3}{C3}, we can conclude $\beta(i,j_2;i,j) \geq \beta(i,j;i,j_1)$.

Finally, we show that any vertical spread is bounded by 1, i.e. $0 \leq j_1 < j_2 \leq n_i$, $-\beta(i,j_2;i,j_1) \leq 1$. For the case $j_1 > 1$, given that any butterfly spread is non-negative, we have $-\beta(i,j_2;i,j_1) \leq -\beta(i,j_1;i,1) \leq -\beta(i,1;i,0) \leq 1$, where the last inequality holds due to the vertical spread constraint \hyperlink{c2}{C2}. If $j_1 = 1$, then $-\beta(i,j_2;i,j_1=1) \leq -\beta(i,1;i,0) \leq 1$. Otherwise $j_1=0$, applying (\ref{eq:claimj1}) by assigning $j_1=0$, $j=j_2$ yields $-\beta(i,j_2;i,j_1=0) \leq -\beta(i,1;i,0) \leq 1$.
\end{proof}

\begin{lemma}\label{lemma:arbitrage2}
If \hyperlink{c2}{C2}, \hyperlink{c4}{C4} and \hyperlink{c5}{C5} are satisfied, then any calendar spread or calendar vertical spread is non-negative.
\end{lemma}

\begin{proof}
We would like to prove that $\forall 0 \leq i_1 < i_2 \leq m$ and $\forall j_1 \in [0, n_{i_1}], ~j_2 \in [0, n_{i_2}]$ where $k^{i^1}_{j^1} \geq k^{i^2}_{j^2}$, we have $c^{i_1}_{j_1} \leq c^{i_2}_{j_2}$.

First consider the calendar spread case when $k^{i_1}_{j_1} = k^{i_2}_{j_2}$. The calendar spread constraint \hyperlink{c4}{C4} immediately leads to $c^{i_1}_{j_1} \leq c^{i_2}_{j_2}$.

Otherwise $k^{i_1}_{j_1} > k^{i_2}_{j_2}$, which implies that $j_1$ must be greater than 0. Given $i_1 \in [1,m]$, $j_1 \in [1, n_{i_1}]$, there must be $k^{i_2}_{j_2} \in [k^{i_1}_{j_1-p-1}, k^{i_1}_{j_1-p})$ for some $p \in [0, j_1-1]$. By the calendar vertical spread constraint \hyperlink{c5}{C5}, we have $c^{i_2}_{j_2} \geq c^{i_1}_{j_1-p}$. In addition, $c^{i_1}_{j_1-p} \geq c^{i_1}_{j_1}$ due to the vertical spread constraint \hyperlink{c2}{C2}. Hence, $c^{i_2}_{j_2} \geq c^{i_1}_{j_1}$.
\end{proof}

\begin{lemma}\label{lemma:arbitrage3}
If \hyperlink{c3}{C3} and \hyperlink{c61}{C6} are satisfied, then any calendar butterfly is non-negative.
\end{lemma}

\begin{proof}
We would like to prove that $\forall i, i_1, i_2 \in [1,m] ~\text{where } i\leq i_1, i\leq i_2$ and $\forall j \in [1, n_i], ~j_1 \in [0, n_{i_1}], ~j_2 \in [0, n_{i_2}] ~\text{where } k^{i^1}_{j^1} < k^{i}_{j} < k^{i^2}_{j^2}$, we have $\beta(i, j; i_1, j_1) \leq \beta(i_2, j_2; i, j)$.

Given $i \in [1,m]$, $j \in [1,n_i]$, it must be that $k^{i^1}_{j^1} \in [k^i_{j-p-1}, k^i_{j-p}]$ for some $p \in [0, j-1]$ and either $k^{i^2}_{j^2} \in [k^i_{j+q}, k^i_{j+q+1}]$ for some $q \in [0, n_i-j-1]$ or $k^{i^2}_{j^2} \in (k^i_{n_i}, \infty)$. See Figure \ref{fig:proofCBS}.

\begin{figure}[h]  
\centering 
\begin{subfigure}[b]{0.48\linewidth}
    \begin{tikzpicture}[scale=1][
   	thick,
    >=stealth']
  	
  	\coordinate (O) at (0,0);

 	\draw[->] (-0.1,0) -- (5.6,0) coordinate[label = {below:$k$}] (xmax);
  	\draw[->] (0,-0.1) -- (0,5) coordinate[label = {right:$c$}] (ymax);	
	
	\node (a) at (.2,4.5) {};
	\node (b) [circle, fill=black, minimum size=3pt, inner sep=0pt] at (.5,3.5) {};
	\node (c) [circle, fill=black, minimum size=3pt, inner sep=0pt] at (1,2.5) {};
	\node (d) [circle, fill=black, minimum size=3pt, inner sep=0pt] at (2,1.5) {};
	\node (de)[rotate=-25] at (2.5,1.25) {...};
	\node (e) [circle, fill=black, minimum size=3pt, inner sep=0pt] at (3,1) {};
	\node (f) [circle, fill=black, minimum size=3pt, inner sep=0pt] at (4,.5) {};
	\node (g) [circle, fill=black, minimum size=3pt, inner sep=0pt, label={above:$(k^{i}_{j}, c^{i}_{j})$}] at (5.1,.3) {};
	\node (h) at (5.5,.25) {};
	
	\draw[gray] (a)--(b);
	\draw[gray] (b)--(c);
	\draw[gray] (c)--(d);
	\draw[gray] (e)--(f);
	\draw[gray] (f)--(g);
	\draw[gray] (g)--(h);
	
	\node (l) [circle, fill=red!60, minimum size=3pt, inner sep=0pt, label={right:$(k^{i_1}_{j_1}, c^{i_1}_{j_1})$}] at (.6,3) {};
	
	\draw[red, densely dotted] (l)--(g);
	\draw[red, densely dotted] (l)--(c);
	\draw[black, densely dotted] (g)--(c);
	\draw[black, densely dotted] (g)--(d);
	\draw[black, densely dotted] (g)--(e);
	
	\draw[dashed] (c) -- (c |- O) node[label = {below:$k^{i}_{j-p}$}] {};
 	\draw[dashed] (d) -- (d |- O) node[label = {below:$k^{i}_{j-p+1}$}] {};
\end{tikzpicture}
    \caption{$k^{i^1}_{j^1} \in [k^i_{j-p-1}, k^i_{j-p}]$ for some $p \in [0, j-1]$} \label{fig:proofCBS1}  
  \end{subfigure}
  \begin{subfigure}[b]{0.48\linewidth}
	\begin{tikzpicture}[scale=1][
   	thick,
    >=stealth']
  	
  	\coordinate (O) at (0,0);

 	\draw[->] (-0.1,0) -- (5.6,0) coordinate[label = {below:$k$}] (xmax);
  	\draw[->] (0,-0.1) -- (0,5) coordinate[label = {right:$c$}] (ymax);	
	
	\node (a) at (.2,4.5) {};
	\node (b) [circle, fill=black, minimum size=3pt, inner sep=0pt, label={right:$(k^{i}_{j}, c^{i}_{j})$}] at (.5,3.5) {};
	\node (c) [circle, fill=black, minimum size=3pt, inner sep=0pt] at (1,2.5) {};
	\node (d) [circle, fill=black, minimum size=3pt, inner sep=0pt] at (2,1.5) {};
	\node (de)[rotate=-25] at (2.5,1.25) {...};
	\node (e) [circle, fill=black, minimum size=3pt, inner sep=0pt] at (3,1) {};
	\node (f) [circle, fill=black, minimum size=3pt, inner sep=0pt] at (4,.5) {};
	\node (g) [circle, fill=black, minimum size=3pt, inner sep=0pt] at (5.1,.3) {};
	\node (h) at (5.5,.25) {};
	
	\draw[gray] (a)--(b);
	\draw[gray] (b)--(c);
	\draw[gray] (c)--(d);
	\draw[gray] (e)--(f);
	\draw[gray] (f)--(g);
	\draw[gray] (g)--(h);
	
	\node (r) [circle, fill=blue!60, minimum size=3pt, inner sep=0pt, label={above:$(k^{i_2}_{j_2}, c^{i_2}_{j_2})$}] at (4.7,.3) {};
	
	\draw[blue, densely dotted] (r)--(b);
	\draw[blue, densely dotted] (r)--(f);
	\draw[black, densely dotted] (b)--(d);
	\draw[black, densely dotted] (b)--(e);
	\draw[black, densely dotted] (b)--(f);
	
	\draw[dashed] (e) -- (e |- O) node[label = {below:$k^{i}_{j+q-1}$}] {};
 	\draw[dashed] (f) -- (f |- O) node[label = {below:$k^{i}_{j+q}$}] {};
\end{tikzpicture}
	\caption{$k^{i^2}_{j^2} \in [k^i_{j+q}, k^i_{j+q+1}]$ for some $q \in [0, n_i-j-1]$} \label{fig:proofCBS2}
  \end{subfigure}
\caption{Locations of $k^{i_1}_{j_1}$ and $k^{i_2}_{j_2}$ relative to $k^i_j$.}
\label{fig:proofCBS}
\end{figure}
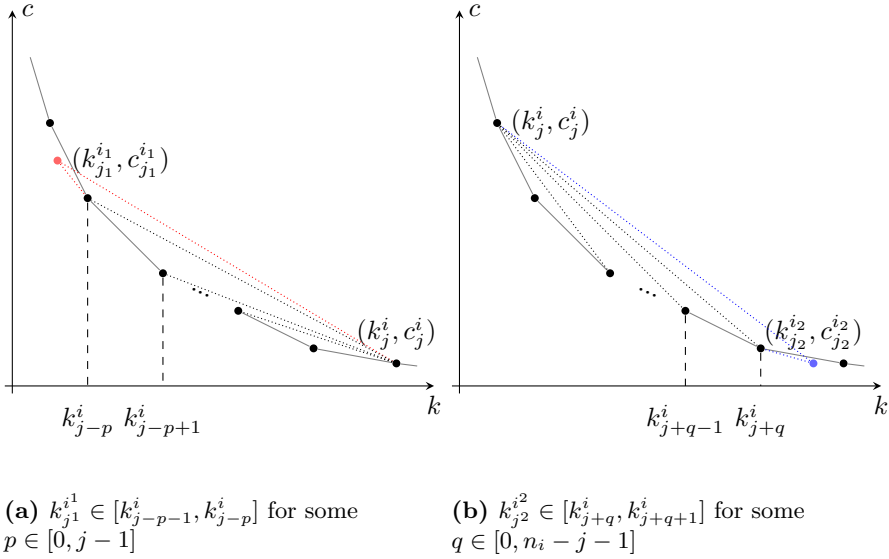

Let us consider the case when $k^{i^2}_{j^2} \leq k^i_{n_i}$ (which implies that $j<n_i$). If $p=q=0$, then by the calendar butterfly relative location constraints \hyperlink{c62}{C6.2} we conclude $\beta(i, j; i_1, j_1) \leq \beta(i_2, j_2; i, j)$. Otherwise, we claim that if $p>0$
\begin{subequations}
\begin{equation}
\beta(i,j;i_1,j_1) \leq \beta(i,j;i,j-p),
\label{eq:claimP1}
\end{equation}
\begin{equation}
\beta(i,j;i,j-p) \leq \beta(i,j;i,j-p+1) \leq \cdots \leq \beta(i,j;i,j-1);
\label{eq:claimP2}
\end{equation}
\label{eq:claimP}
\end{subequations}
and similarly if $q>0$
\begin{subequations}
\begin{equation}
\beta(i_2,j_2;i,j) \geq \beta(i,j+q;i,j),
\label{eq:claimQ1}
\end{equation}
\begin{equation}
\beta(i,j+q;i,j) \geq \beta(i,j+q-1;i,j) \geq \cdots \geq \beta(i,j+1;i,j).
\label{eq:claimQ2}
\end{equation}
\label{eq:claimQ}
\end{subequations}
We will show the proof for the four claims later. If $p>0$ and $q>0$, the four claims and the vertical butterfly constraint \hyperlink{c3}{C3} lead to the stated result. If $p>0$ but $q=0$, then (\ref{eq:claimP}) and the calendar butterfly absolute location convexity constraint \hyperlink{c61}{C6.1} lead to the stated result. If $p=0$ but $q>0$, then (\ref{eq:claimQ}) and \hyperlink{c61}{C6.1} lead to the stated result.

Next we would like to prove the claims (\ref{eq:claimP}) and (\ref{eq:claimQ}). First of all, (\ref{eq:claimP2}) and (\ref{eq:claimQ2}) hold because of the convexity of the set of points $\{(k^i_l, c^i_l)\}_{l\in[j-p,j+q]}$ resulted from the vertical butterfly constraint \hyperlink{c3}{C3}. The calendar butterfly absolute location convexity constraint \hyperlink{c61}{C6.1} results in $\beta(i,j-p;i_1,j_1) \leq \beta(i,j-p+1;i,j-p)$. In addition, the vertical butterfly constraint \hyperlink{c3}{C3} results in $\beta(i,j;i,j-p+1) \geq \beta(i,j-p+1;i,j-p)$, then
\begin{equation*}
\begin{split}
c^i_{j} & \geq c^i_{j-p+1} + (k^i_{j}-k^i_{j-p+1}) \frac{c^i_{j-p+1}-c^i_{j-p}}{k^i_{j-p+1}-k^i_{j-p}} \\
& = c^i_{j-p+1} + \left[ (k^i_{j}-k_{j-p}) - (k^i_{j-p+1}-k_{j-p}) \right] \frac{c^i_{j-p+1}-c^i_{j-p}}{k^i_{j-p+1}-k^i_{j-p}} \\
& = c^i_{j-p} + (k^i_{j}-k_{j-p}) \beta(i,j-p+1;i,j-p).
\end{split}
\end{equation*}
Hence $\beta(i,j,i,j-p) \geq \beta(i,j-p+1,i,j-p) \geq \beta(i,j-p,i_1,j_1)$. Then
\begin{equation*}
\begin{split}
c^{i_1}_{j_1} & \leq c^i_{j-p} + (k^{i_1}_{j_1} - k^i_{j-p}) \frac{c^i_j-c^i_{j-p}}{k^i_j-k^i_{j-p}} = c^i_{j-p} + \left[ (k^{i_1}_{j_1}-k^i_j )-( k^i_{j-p}-k^i_j) \right]  \frac{c^i_j-c^i_{j-p}}{k^i_j-k^i_{j-p}} \\
& = c^i_j + (k^{i_1}_{j_1}-k^i_j ) \beta(i,j;i,j-p),
\end{split}
\end{equation*}
which indicates that $\beta(i,j;i_1,j_1) \leq \beta(i,j;i,j-p)$, i.e. (\ref{eq:claimP1}).

The calendar butterfly spread absolute location convexity constraint \hyperlink{c61}{C6.1} results in $\beta(i_2,j_2;i,j+q) \geq \beta(i,j+q;i,j+q-1)$. In addition, the vertical butterfly constraint \hyperlink{c3}{C3} results in $\beta(i,j+q-1;i,j) \leq \beta(i,j+q;i,j+q-1)$, then
\begin{equation*}
\begin{split}
-c^i_j & \leq - c^i_{j+q-1} + (k^i_{j+q-1}-k^i_j) \frac{c^i_{j+q}-c^i_{j+q-1}}{k^i_{j+q}-k^i_{j+q-1}} \\
& = -c^i_{j+q-1} + \left[ (k^i_{j+q-1}-k^i_{j+q}) - (k^i_j-k^i_{j+q}) \right] \frac{c^i_{j+q}-c^i_{j+q-1}}{k^i_{j+q}-k^i_{j+q-1}} \\
& = -c^i_{j+q} + (k^i_{j+q}-k^i_j) \beta(i,j+q-1;i,j+q).
\end{split}
\end{equation*}
Hence $\beta(i,j+q;i,j) \leq \beta(i,j+q;i,j+q-1) \leq \beta(i_2,j_2;i,j+q)$. Then
\begin{equation*}
\begin{split}
c^{i_2}_{j_2} & \geq c^i_{j+q} + (k^{i_2}_{j_2} - k^i_{j+q}) \frac{c^i_j-c^i_{j+q}}{k^i_j-k^i_{j+q}} = c^i_{j+q} + \left[ (k^{i_2}_{j_2}-k^i_j) - (k^i_{j+q}-k^i_j) \right] \frac{c^i_j-c^i_{j+q}}{k^i_j-k^i_{j+q}} \\
& = c^i_j + (k^{i_2}_{j_2}-k^i_j) \beta(i,j;i,j+q),
\end{split}
\end{equation*}
which indicates that $\beta(i_2,j_2;i,j) \geq \beta(i,j+q;i,j)$, i.e. (\ref{eq:claimQ1}).

Now we consider the case when $k^{i_2}_{j_2} > k^i_{n_i}$. The same proof as above applies if we let $q=n_i-j$ and allow $j$ to take the value $n_i$.

\end{proof}

\section*{Acknowledgements}
This publication is based on work supported by the EPSRC Centre for Doctoral Training in Industrially Focused Mathematical Modelling (EP/L015803/1) in collaboration with CME Group. We thank Florian Huchede, Director of Quantitative Risk Management, and other colleagues at CME Group for providing valuable data access, suggestions from the business perspective, and continued support. 

Samuel Cohen and Christoph Reisinger acknowledge the support of the Oxford-Man Institute for Quantitative Finance, and Samuel Cohen also acknowledges the support of the Alan Turing Institute under the Engineering and Physical Sciences Research Council grant EP/N510129/1.

\bibliographystyle{abbrv}
\bibliography{reference}   

\end{document}